\documentclass[a4paper,11pt]{article} 
  
\usepackage[english]{babel}
\usepackage[a4paper]{geometry}
\usepackage{amsmath}
\usepackage{amsthm}
\usepackage{amssymb}
\usepackage{bbm}
\usepackage{color}
\usepackage{enumerate}

\usepackage{hyperref}         
\usepackage{tikz}
\usetikzlibrary{fit}
\usetikzlibrary{shapes.geometric}
\usetikzlibrary{decorations.pathmorphing}

\tikzset{
point/.style={circle,fill=black,inner sep=1pt},
vertex/.style={circle,fill=black,inner sep=1.5pt},   
bvertex/.style={circle,fill=black,inner sep=2.8pt},
Bvertex/.style={circle,fill=black,inner sep=4pt}, 
specialEP/.style={rectangle,fill=white,draw,inner sep=3pt},  
whitevex/.style={circle,fill=white,draw, inner sep=2pt},
linelabel/.style={sloped,above,very near start, inner sep=1pt,execute at begin node=$\scriptstyle,execute at end node=$},
baseline=(current  bounding  box.center),doubled/.style={double distance= 1pt,line width=1.5pt},
th/.style={line width=0.5 pt, gray},  
med/.style={line width=1 pt}  
}

\xdefinecolor{myred}{rgb}{0.7,0,0.2} 
\xdefinecolor{mypurple}{rgb}{0.6,0.2,0.4} 
\xdefinecolor{myblue}{rgb}{0.259,0.2,0.6}   
\xdefinecolor{myorange}{rgb}{1,0.6,0.2}   
\definecolor{orange}{rgb}{1,0.5,0}
\xdefinecolor{purpleish}{cmyk}{0.75,0.75,0,0}

\def\bR{\mathbb{R}}

\def\bN{\mathbb{N}}
\def\NN{\mathbb{N}}
\def\bZ{\mathbb{Z}}
\def\cC{\mathcal{C}}

\def\cV{\mathcal{V}}
\def\cO{\mathcal{O}}
\def\cF{\mathcal{F}}
\def\cG{\mathcal{G}}
\def\cL{\mathcal{L}}

\def\cN{\mathcal{N}}
\def\cE{\mathcal{E}}
\def\cK{\mathcal{K}}
\def\cH{\mathcal{H}}

\def\eps{\varepsilon}
\def\ph{\varphi}

\def\wt{\widetilde}

\def\indic{\hbox{\raise-2pt \hbox{\indbf 1}}}

\let\==\equiv

\let\io=\infty
\let\0=\noindent

\def\*{{\hfill\break\null\hfill\break}}

\def\bmedia#1{{\bigl\langle#1\bigr\rangle}}

\def\tende#1{\,\vtop{\ialign{##\crcr\rightarrowfill\crcr
             \noalign{\kern-1pt\nointerlineskip}
             \hskip3.pt${\scriptstyle #1}$\hskip3.pt\crcr}}\,}
\def\otto{\,{\kern-1.truept\leftarrow\kern-5.truept\to\kern-1.truept}\,}

\newtheorem{theorem}{Theorem}[section]  

\newtheorem{prop}[theorem]{Proposition}
\newtheorem{lemma}[theorem]{Lemma}

\numberwithin{equation}{section}

\def\tl#1{{\tilde{#1}}}
\def\be{\begin{equation}}
\def\ee{\end{equation}}

\newcommand{\hc}{\mbox{h.c.}}

\let\a=\alpha     \let\g=\gamma     \let\d=\delta     \let\e=\varepsilon
        \let\k=\kappa     \let\l=\lambda
                  \let\p=\pi        \let\r=\rho
\let\s=\sigma \let\t=\tau         \let\ph=\varphi   
   \let\o=\omega     
 \let\D=\Delta       \let\L=\Lambda    
             
\let\O=\Omega

\definecolor{lightblue}{rgb}{0, 0.33, 0.71}

\def\aa{\mathfrak{a}}
\def \blue#1 {\textcolor{blue}{#1}}

\title{A new second order upper bound for the \\ ground state energy of dilute Bose gases}

\author{Giulia Basti\footnote{Gran Sasso Science Institute, Viale Francesco Crispi 7, 67100 L'Aquila, Italy}\;, Serena Cenatiempo$^\ast$, Benjamin Schlein\footnote{Institute of Mathematics, University of Zurich, Winterthurerstrasse 190, 8057 Zurich.}}

\begin{document}

\maketitle

\begin{abstract} We establish an upper bound for the ground state energy per unit volume of a dilute Bose gas in the thermodynamic limit, capturing the correct second order term, as predicted by the Lee-Huang-Yang formula. This result has been first established in \cite{YY} by H.-T. Yau and J. Yin. Our proof, which applies to repulsive and compactly supported $V \in L^3 (\bR^3)$, gives better rates  and, in our opinion, is substantially simpler. 
\end{abstract}

\section{Introduction and main result}

We consider $N$ bosons in a finite box $\L_L=[-\tfrac L 2, \tfrac L 2]^3 \subset \bR^3$, interacting via a two-body non negative, radial, compactly supported potential $V$ with scattering length $\aa$. The Hamilton operator has the form
\begin{equation} \label{eq:HN-0}
H_L = -\sum_{i=1}^N \D_i+\sum_{1 \leq i<j \leq N}V(x_i-x_j)
\end{equation}
and acts on the Hilbert space $L^2_s(\L_L^N)$, the subspace of $L^2(\L_L^N)$ consisting of functions that are symmetric with respect to permutations of the $N$ particles (we use here units with particle mass $m=1/2$ and $\hbar =1$). We assume Dirichlet boundary conditions and denote by $E(N,L)$ the corresponding ground state energy. We are interested in the energy per unit volume in the thermodynamic limit, defined by 
\begin{equation} \label{eq:e-rho}
e(\r)=\lim_{\substack{N,L\to+\infty\\ \r=N/L^3}} \frac{E(N,L)}{L^3}\,.
\end{equation}

Bogoliubov \cite{B} and later, in more explicit terms, Lee-Huang-Yang \cite{LHY} predicted that, in the dilute limit $\r \aa^3 \ll 1$, the specific ground state energy \eqref{eq:e-rho} is so that 
\be \label{eq:LHY}
e(\r)=  4 \pi \aa \r^2  \left[ 1+ \frac{128}{15\sqrt \pi} (\r \aa^3)^{1/2}+ o( (\r \aa^3)^{1/2}) \right]
\ee
In particular, up to lower order corrections, it only depends on the interaction potential through the scattering length $\aa$. 

The validity of the leading term on the r.h.s. of (\ref{eq:LHY}) was established by Dyson, who obtained an upper bound in \cite{Dy}, and by Lieb-Yngvason, who proved the matching lower bound in \cite{LY}. 
A rigorous upper bound with the correct second order contribution was first derived in \cite{YY} by Yau-Yin for regular potentials, improving a previous estimate from \cite{ESY}, which only recovered the correct formula (as an upper bound) in the limit of weak coupling. The approach of \cite{YY} has been reviewed and adapted to a grand canonical setting in \cite{Aaen}. As for the lower bound, preliminary results have been obtained in \cite{GiuS} and \cite{BriS}, where (\ref{eq:LHY}) was shown in particular regimes, where the potential scales with the density $\rho$. Finally, a rigorous lower bound matching (\ref{eq:LHY}) has been recently obtained, for $L^1$ potentials, by Fournais-Solovej in \cite{FS} (for hard core potentials, the best available lower bound is given in \cite{BrFS} and it matches (\ref{eq:LHY}), up  to corrections of the same size as the second order term). 

Our goal, in this paper, is to show a new upper bound for (\ref{eq:LHY}). With respect to the upper bound established in \cite{YY}, our result holds for a larger class of potentials (in \cite{YY}, the upper bound is proven for smooth potentials), it gives a better rate (although still far from optimal) and, most importantly in our opinion, it relies on a simpler proof.    
\begin{theorem} \label{thm:main}
Let $V \in L^{3}(\bR^3)$ be non-negative, radially symmetric, with $\mathrm{supp}(V)\subset B_R(0)$ and scattering length $\aa\leq R$ . Then, the specific ground state energy $e(\rho)$ of the Hamilton operator $H_L$  defined in \eqref{eq:HN-0} satisfies
\begin{equation}\label{eq:main}
e(\r)\leq  4\pi \r^2 \aa \Big[ 1 + \frac{128}{15 \sqrt \pi} (\r \aa^3)^{1/2} \Big] + C \r^{5/2+ 1/10}
\end{equation} 
for some $C>0$ and for $\r$ small enough.
\end{theorem}

{\it Remark:} since Dirichlet boundary conditions lead to the largest energy, the upper bound (\ref{eq:main}) holds in fact for arbitrary boundary conditions. 

\medskip

{\it Remark:} at the cost of a longer proof, we could improve the bound on the error, up to the order $\rho^{5/2+2/9}$ (this is the rate determined by Lemma \ref{lm:cubic}). 

\medskip

The proof of \ref{eq:main} is based on the construction of an appropriate trial state. However, we do not construct directly a trial state in $L^2_s (\L_L^N)$ for the Hamiltonian (\ref{eq:HN-0}). Instead, to simplify the analysis, it is very convenient to 1) consider smaller boxes (rather than letting $N, L \to \infty$ first and considering small $\r$ at the end, we will consider a diagonal limit, with $L = \rho^{-\g}$, for some $\g > 1$), 2) work with periodic rather than Dirichlet boundary conditions and 3) work in the grand-canonical, considering states with variable number of particles, rather than the canonical setting. In other words, our trial state will be defined on the bosonic Fock space 
\[\label{eq:Fock} 
\cF(\L_{ L})= \bigoplus_{n \geq 0} L^2_s (\Lambda_L^{n}) = \bigoplus_{n \geq 0} L^2 (\Lambda_L)^{\otimes_s n} \]
where $L^2_s (\Lambda_L^{n})$ is the subspace of $L^2 (\Lambda_L^n)$ consisting of wave functions that are symmetric w.r.t. permutations. On $\cF(\L_{ L})$, we consider the number of particles operator $\cN$ defined through $(\cN \psi)^{(n)}= n \psi^{(n)}$. Moreover, we introduce the Hamiltonian 
operator $\cH$, setting  
\be \label{eq:cH}
(\cH \psi)^{(n)}= \cH^{(n)} \psi^{(n)}
\ee with
\[ \label{eq:cHn}
\cH^{(n)} = \sum_{j=1}^n -\D_{x_j} + \sum_{1 \leq i<j \leq n}  V(x_i-x_j) 
\]
imposing now (in contrast to what we did in (\ref{eq:HN-0})) periodic boundary conditions. The upper bound for the energy of (\ref{eq:cH}) will then imply Theorem \ref{thm:main} thanks to the following localization result, whose standard proof \cite{R, YY, Aaen} is discussed for completeness in Appendix \ref{app:localization}. 

\begin{prop} \label{prop:localization}
Let  $e(\r)$ be defined as in \eqref{eq:e-rho}, with Dirichlet boundary conditions. Let $R < \ell < L$, with $R$ the radius of the support of the potential $V$, as defined in Theorem \ref{thm:main}. 
Then, for any normalized $\Psi_{L} \in \cF (\L_{L})$ satisfying periodic boundary conditions and such that 
\be \label{eq:loc-assump}
 \langle \Psi_{L}, \cN \Psi_{L} \rangle  \geq \rho(1+c' \r) (L + 2 \ell + R)^3 \,,\quad  \langle \Psi_{L}, \cN^2 \Psi_{L} \rangle  \leq C' \rho^{2} (L + 2 \ell + R)^6
\ee
for some $c' , C' > 0$. Then we have 
\be \label{eq:localization}
e(\r) \leq   \frac{\langle \Psi_{L}, \cH \Psi_{L} \rangle }{L^3} + \frac{C}{L^{4} \ell} \langle \Psi_{L}, \cN \Psi_{L} \rangle 
\ee 
for a universal constant $C > 0$.
\end{prop}

The bulk of the paper contains the proof of the following proposition, establishing the existence of a trial state with the correct energy per unit volume and the correct expected number of particles, on boxes of size $L = \tl{\rho}^{-\g}$. We use here the notation $\tl{\rho}$ for the density to stress the fact the upper bound (\ref{eq:enL}) will be inserted in (\ref{eq:localization}) to prove an upper bound for the specific ground state energy $e (\r)$, for a slightly different density $\r < \tl{\rho}$ (to make up for the corrections on the r.h.s. of (\ref{eq:loc-assump})). 
\begin{prop} \label{prop:energy-trial}  
 As in Theorem \ref{thm:main} assume that $V \in L^3 (\bR^3)$ is non-negative, radially symmetric with $\text{supp } V \subset B_R (0)$ and scattering length $\frak{a} \leq R$. For $\gamma > 1$ and $\tl{\rho} > 0$ let $L = \tl{\rho}^{-\g}$. Then, for every $0< \eps < 1/4$, there exists $\Psi_{\tl{\rho}} \in \cF (\L_L)$ satisfying periodic boundary conditions such that 
 \begin{equation}\label{eq:tlrho-def} 
\langle \Psi_{\tl{\rho}} , \cN \Psi_{\tl{\rho}} \rangle \geq \tl{\rho} L^3 \,, \qquad  \langle \Psi_{\tl{\rho}} , \cN^2 \Psi_{\tl{\rho}} \rangle \leq C \tl \r^{2} L^6 \end{equation} 
and 
\begin{equation}\label{eq:enL}
	\frac{\bmedia{\Psi_{\tl\r},\cH \Psi_{\tl \r}}}{ L^3} \leq 4\pi\aa \tl\r^{\,2}\bigg(1+\frac{128}{15\sqrt{\pi}}(\aa^3 \tl \r)^{1/2}\bigg)+ \cE \,,
\end{equation} 
with 
\[ \cE \leq C \,\tl\r^{\,5/2} \cdot \max \{\tl{\rho}^\eps , \, \tl{\rho}^{4-3\g -6\eps} , \, \tl{\rho}^{9/4-3\g/2-3\eps} \}\,. \]
\end{prop}

{\it Remark:} the condition $\gamma > 1$ is needed to make sure that the localization error in (\ref{eq:localization}) is negligible. While we will chose $\gamma = 11/10$ to optimize the rate, our analysis allows us to take any $1 < \gamma < 4/3$. With a longer proof, our techniques could be extended to all $1 < \gamma < 5/3$. This suggests that our trial state captures the correct correlations of the ground state, up to length scales of the order $\rho^{-5/3}$. 

\medskip

With Prop. \ref{prop:localization} and Prop. \ref{prop:energy-trial} we can prove Theorem \ref{thm:main}. 

\begin{proof}[Proof of Theorem \ref{thm:main}] 
For given $\rho > 0$, we would like to choose $\tl{\rho}$, or equivalently $L = \tl{\rho}^{-\g}$, so that (\ref{eq:tlrho-def}) implies (\ref{eq:loc-assump}). Fixing $c' > 0$ and $\ell = L^\alpha$, for some $\alpha \in (0;1)$, this leads to the implicit equation 
\begin{equation} \label{eq:impl1} L = \tl{\rho}^{-\g} = \left[ \rho (1+ c' \r) (1 + 2 L^{\alpha-1} + R L^{-1})^3 \right]^{-\g}\,. \end{equation} 
Setting $L = \big(\rho (1+c' \r) \big)^{-\g} x$, we rewrite (\ref{eq:impl1}) as 
\[ x = \big(1 + 2 \, \big(\rho(1+c' \r) \big)^{\g (1-\a)} / x^{1-\a} + R \, \big(\rho(1+c' \r) 
\big)^\g /x \big)^{-3\g} \]
and we conclude that the existence of a solution $L = L(\rho)$ of (\ref{eq:impl1}) follows from the implicit function theorem, if $\rho > 0$ is small enough (the solution stems from $x=1$ for $\rho = 0$). By construction $L = \tl{\rho}^{-\g}$, with \[ \tl{\rho} = \rho  (1+c' \r) (1 + 2 \tl{\rho}^{\g (1-\a)} + R \tl{\rho}^\g)^3\]  and thus 
\begin{equation}\label{eq:rho-rhotl} 
\rho \leq \tl{\rho}  \leq \rho (1 + C \rho + C \rho^{\g (1-\a)}) \,.
\end{equation} 
From Prop. \ref{prop:energy-trial}, we find $\Psi_{\tl{\rho}} \in \cF (\L_L)$ such that (\ref{eq:tlrho-def}) and (\ref{eq:enL}) hold true. In particular, (\ref{eq:tlrho-def}) implies (\ref{eq:loc-assump}) (with $\ell = L^\alpha$, $C'=C$). 
Thus, from Prop. \ref{prop:localization} we conclude  
\[ e(\rho) \leq \frac{\langle \Psi_{\tl{\rho}}, \cH \Psi_{\tl{\rho}} \rangle}{L^3} + \frac{C}{L^4 \ell} \langle \Psi_{\tl{\rho}} , \cN \Psi_{\tl{\rho}} \rangle \]
Inserting (\ref{eq:enL}) and (\ref{eq:tlrho-def}), we obtain (since (\ref{eq:tlrho-def}) also implies that $\langle \Psi_{\tl{\rho}} , \cN \Psi_{\tl{\rho}} \rangle \leq C \wt{\rho} L^3$) 
\[ e (\rho) \leq 4 \pi \frak{a} \tl{\rho}^2 \left[ 1 + \frac{128}{15\sqrt{\pi}} (\frak{a}^3 \tl{\rho})^{1/2} \right] + C \tl{\rho}^{1+\g (1+\a)} + C \tl\r^{\,5/2} \cdot \max \{\tl{\rho}^\eps , \, \tl{\rho}^{4-3\g -6\eps} , \, \tl{\rho}^{9/4-3\g/2-3\eps} \}\,. \]
With (\ref{eq:rho-rhotl}), we conclude that
\[ \begin{split} e (\rho) \leq \; &4 \pi \frak{a} \rho^2 \left[ 1 + \frac{128}{15\sqrt{\pi}} (\frak{a}^3 \rho )^{1/2} \right] \\ &+ C \r^{5/2} \cdot \max \{ \rho^{\g (1-\a)-1/2} , \,  \rho^{\g (1+\a) -3/2} , \,   \rho^\eps , \, \rho^{4-3\g -6\eps} , \, \rho^{9/4-3\g/2-3\eps} \}  \end{split} \]
where we neglected errors of order $C \r^3$, which are subleading being $\e \in (0; 1/4)$.

Comparing the first two errors, we choose $\alpha = 1/(2\g)$ . Comparing instead third and fourth errors, we set $\eps = (4-3\g)/7$ (both choices are consistent with the conditions $\a \in (0;1)$ and $\e \in (0;1/4)$, because $\g > 1$). Since, with these choices, the last error is of smaller order, we obtain 
 \[ e (\rho) \leq \; 4 \pi \frak{a} \rho^2 \left[ 1 + \frac{128}{15\sqrt{\pi}} (\frak{a}^3 \rho )^{1/2} \right]  + C \r^{5/2} \cdot \max \{ \rho^{\g-1}, \rho^{(4-3\g)/7} \}\,. \]
 Choosing $\gamma = 11/10$, we find (\ref{eq:main}). 
\end{proof}

The proof of Prop. \ref{prop:energy-trial} occupies the rest of the paper (excluding Appendix \ref{app:localization}, where we show Prop. \ref{prop:localization}). In Section \ref{sec:trial} we 
define our trial state. To this end, we will start with a coherent state describing the Bose-Einstein condensate. Similarly as in \cite{GA,ESY}, we will then act on the coherent state with a Bogoliubov transformation to add the expected correlation structure. Finally, we will apply the exponential of a 
cubic expression in creation operators. While the Bogoliubov transformation creates pairs of excitations with opposite momenta $p, -p$, the cubic operator creates three excitations at a time, two with large momenta $r+v, -r$  and one with low momentum $v$. This last step is essential, since, as follows from \cite{ESY,NRS}, quasi-free states cannot approximate the ground state energy to the precision of (\ref{eq:LHY}). We remark that the idea of creating triples of excitations originally appeared in the work of Yau-Yin \cite{YY}  (a brief comparison with the trial state of \cite{YY} can be found after the precise definition of our trial state in \eqref{eq:trial}). Recently, it has been also applied to establish the validity of Bogoliubov theory in the Gross-Pitaevskii regime in \cite{BBCS3,BBCS4}; while our approach is inspired by these papers, we need here new tools to deal with the large boxes considered in Prop.\ref{prop:energy-trial} (a simple computation shows that the Gross-Pitaevskii regime corresponds to the exponent $\gamma = 1/2$; to control localization errors, we need instead to choose $\gamma > 1$). 
In Section \ref{sec:energy}, we combine the contributions to the energy of the trial state arising from the conjugation with the Bogoliubov transformation and from the action of the cubic phase, proving the desired upper bound. In Section \ref{sec:Bog} and Section \ref{sec:cubicconj}, we prove technical bounds which allow us to identify the leading contributions collected in Sect.  \ref{sec:energy}.  \\

\thanks{{\it Acknowledgment.}  We are grateful to C. Boccato and S. Fournais for valuable discussions. G.B. and S.C. gratefully acknowledge the support from the GNFM Gruppo Nazionale
per la Fisica Matematica. 
B. S. gratefully acknowledges partial support from the NCCR SwissMAP, from the Swiss National Science Foundation through the Grant ``Dynamical and energetic properties of Bose-Einstein condensates'' and from the European Research Council through the ERC-AdG CLaQS.}

\section{Setting and trial state} \label{sec:trial} 

To show Prop. \ref{prop:energy-trial}, we find it convenient to work with rescaled variables. We consider the transformation $x_j \to x_j / L$, and, motivated by the choice $L = \tl{\rho}^{-\g}$ in Prop. \ref{prop:energy-trial}, we set $N = \tl{\rho}^{1-3\g}$ (we will look for trial states with expected number of particles close to $N$ to make sure that (\ref{eq:tlrho-def}) holds true). It follows that the Hamiltonian (\ref{eq:cH}) is unitarily equivalent to the operator $L^{-2} \cH_N = \tl{\rho}^{2\g} \cH_N$, with $\cH_N$ acting on the Fock space $\cF (\L)$ defined over the unit box $\L = \Lambda_1 = [-1/2 ; 1/2 ]^3$ (with periodic boundary conditions) so that $(\cH_N \Psi)^{(n)}  = \cH_N^{(n)} \Psi^{(n)}$, with 
\[ \cH_N^{(n)} = \sum_{j=1}^n -\Delta_{x_j} + \sum_{1 \leq i,j \leq n} N^{2-2\k} V (N^{1-\k} (x_i - x_j)) \]
and $\kappa = (2\g - 1)/ (3\g -1)$. The assumption $\g > 1$ in Prop. \ref{prop:energy-trial} allows us to restrict our attention to $\kappa \in (1/2 ; 2/3)$. 

For any momentum  $p \in \Lambda^* = 2\pi \bZ^3$, we introduce on the Fock space $\cF (\L)  = \bigoplus_{n \geq 0} L^2_s (\L^{n})$, the operators $a_p^*, a_p$, creating and, respectively, annihilating a particle with momentum $p$. Creation and annihilation operators satisfy the canonical commutation relations
\begin{equation} \label{eq:CCR} \left[ a_p , a_q^* \right] = \delta_{pq} , \qquad \left[ a_p , a_q  \right] = \left[ a^*_p, a_q^* \right] = 0. \end{equation} 
On $\cF(\L)$, we define the number of particles operator $\cN = \sum_{p \in \Lambda^*} a_p^* a_p$. Expressed in terms of creation and annihilation operators, the Hamiltonian $\cH_N$ takes the form 
\be \label{eq:HN-Fock}
	\cH_N=\sum_{p\in\L^*}p^2a_p^*a_p+\frac{1}{2N^{1-\k}} \sum_{p,q,r\in \L^*} \widehat{V} (r/N^{1-\k}) \, a_{p+r}^*a_q^*a_{q+r}a_{p}.
\ee

We construct now our trial state. To generate a condensate, we use a Weyl operator 
\be \label{eq:defW}
W_{N_0}=\exp \big[ \sqrt{N_0} a_0^* - \sqrt{N_0}a_0 \big]
\ee
with a parameter $N_0$ to be specified later on. While $W_{N_0}$ leaves $a_p, a_p^*$ invariant, 
for all $p \in \Lambda^* \backslash \{ 0 \}$, it produces shifts of $a_0, a_0^*$; in other words
\begin{equation}\label{eq:actW} W_{N_0}^* a_0 \, W_{N_0} = a_0 + \sqrt{N_0} , \qquad W_{N_0}^* a_0^* \, W_{N_0} = a_0^* + \sqrt{N_0}. \end{equation} 
When acting on the vacuum vector $\Omega = \{ 1, 0, \dots \}$, (\ref{eq:defW}) generates a coherent state in the zero-momentum mode $\ph_0 (x) \equiv 1$, with expected number of particles $N_0$.  

It turns out, however, that the coherent state does not approximate the ground state energy, not even to leading order. To get closer to the ground state energy, it is crucial to add correlations among particles. To this end, we fix $0 < \ell < 1/2$ and we consider the lowest energy solution $f_{\ell}$ of the Neumann problem 
\begin{equation}\label{eq:scatl} \left[ -\Delta + \frac{1}{2} V \right] f_{\ell} = \lambda_{\ell} f_{\ell} \end{equation}
on the  ball $|x| \leq N^{1-\k}\ell$, with the normalization $f_\ell (x) = 1$ if $|x| =N^{1-\k} \ell$. Furthermore, by rescaling, we define $f_N (x) := f_\ell \big( N^{1-\k}x\big)$ for $|x| \leq \ell$. We extend $f_N$ to a function on $\Lambda$, by fixing $f_N (x) = 1$, for all $x \in \Lambda$, with $|x| > \ell$. Then 
\begin{equation}\label{eq:scatlN}
 \left[ -\Delta + \frac12 N^{2-2\k} V ( N^{1-\k}x ) \right] f_{N} (x) = N^{2-2\k} \lambda_\ell f_{N} (x) \chi_\ell(x)
\end{equation}
for all $x \in \L$, where $\chi_\ell$ denotes the characteristic function of the ball of radius $\ell$. We denote by $ \widehat{f}_{N} (p)$ the Fourier coefficients of the function $f_{N}$, for $p \in \L^*$. We also define $w_\ell (x) = 1 - f_\ell (x)$ (with $w_\ell (x) = 0$ for $|x| > N^{1-\kappa} \ell$) and its rescaled version $w_N : \Lambda \to \bR$ through $w_N (x) = w_\ell (N^{1-\kappa} x) = 1 - f_N (x)$. The Fourier coefficients of $w_{N}$ are given by  
\[  \widehat{w}_{N} (p) = \int_{\Lambda} w_\ell ( N^{1-\k} x) e^{-i p \cdot x} dx =\frac1{N^{3-3\k}}\, \widehat{w}_\ell \big(p/N^{1-\k}\big) \]
where $\widehat{w}_\ell (k)$ denotes the  Fourier transform of the (compactly supported) function $w_\ell$. 
Some important properties of the solution of the eigenvalue problem (\ref{eq:scatl}) are summarized in the following lemma, whose proof can be found in \cite[Appendix A]{BBCS3} (replacing $N\in \bN$ by $N^{1-\k}$). 
\begin{lemma} \label{sceqlemma}
Let $V \in L^3 (\bR^3)$ be non-negative, compactly supported and spherically symmetric. Fix $\ell > 0$ and let $f_\ell$ denote the solution of \eqref{eq:scatl}. For $N\in\NN$ large enough the following properties hold true. 
\begin{enumerate}
\item [i)] We have 
\begin{equation*}\label{eq:lambdaell} 
 \bigg| \lambda_\ell - \frac{3\aa }{N^{3-3\k}\ell^3} \bigg| \leq  \frac{1}{N^{3-3\k}\ell^3} \frac{C \aa^2}{\ell N^{1-\k}}\,.
\end{equation*}
\item[ii)] We have $0\leq f_\ell, w_\ell\leq1$. Moreover there exists a constant $C > 0$ such that  
\begin{equation*} \label{eq:Vfa0} 
\left|  \int  V(x) f_\ell (x) dx - 8\pi \aa   \right| \leq \frac{C \aa^2 }{\ell N^{1-\k}}.
\end{equation*}
\item[iii)] There exists a constant $C>0 $ such that, for all $x \in \bR^3$, 
	\begin{equation*}\label{3.0.scbounds1} 
	w_\ell(x)\leq \frac{C}{|x|+1} \quad\text{ and }\quad |\nabla w_\ell(x)|\leq \frac{C }{x^2+1}. 
	\end{equation*}
\item[iv)] There exists a constant $C > 0$ such that, for all $p \in \bR^3$, 
\[ |\widehat{w}_{N} (p)| \leq  \frac{C}{ N^{1-\k}p^2}  \, . \]
\end{enumerate}        
\end{lemma}

We consider the coefficients $\eta: \L^* \to \bR$ defined through
\begin{equation}\label{eq:defeta}
\eta_p = -N \widehat{w}_{N} (p) = - \frac{N^\k}{N^{2-2\k}} \,\widehat{\o}_\ell(p/N^{1-\k}).
\end{equation}
Lemma \ref{sceqlemma}  implies that 
\be \label{eq:modetap}
|\eta_p| \leq \frac{C N^\k}{p^2}
\ee
for all $p \in \L_+^*=2\pi \bZ^3 \backslash \{0\}$, and  for some constant $C>0$ independent of $N\in\NN$ (for $N\in\NN$ large enough). {F}rom (\ref{eq:scatlN}), we find the relation 
\begin{equation}\label{eq:scatt}
p^2 \eta_p + \frac{N^\k}{2}\widehat{V} (p/N^{1-\k}) + \frac{1}{2N} \sum_{q \in \Lambda^*} N^\k\widehat{V} ((p-q) /N^{1-\k})\eta_q = N^{3-2\k} \l_{\ell} (\widehat{\chi}_\ell * \widehat{f}_{N}) (p)\,. \end{equation}
From Lemma \ref{sceqlemma}, part iii), we also obtain
\begin{equation}\label{eq:wteta0}  |\eta_0| \leq N^{3-2\k} \int_{\bR^3} w_\ell  (x) dx \leq C N^\k \,. \end{equation}

The coefficients $\eta_p$ will be used to model, through a Bogoliubov transformation, short-distance correlations among particles. To reach this goal, it is enough to act on momenta $|p| \gg N^{\k/2}$. On low momenta, the Bogoliubov transformation is needed to diagonalize the (renormalized) quadratic part of the Hamiltonian. For $\eps > 0$ small enough we define therefore the set 
\be \begin{split} \label{eq:PL-PH}
P_L & \; = \Big\{ p \in \L^*_+ :  |p| \leq N^{\k/2+\eps}  \Big\}\,,
\end{split}\ee
of low momenta (the condition on $\eps$ makes sure that the two sets are disjoint). We will denote its complement by $P_L^c = \Lambda_+^* \backslash P_L$.
For $p \in \L^*_+$ we set 
\[ \label{eq:defnup}
\nu_p = \t_p  \chi(p \in P_L) + \eta_p \chi(p \in P_L^c)
\]
with $\eta_p$ defined in \eqref{eq:defeta} and $\t_p \in \bR$ defined by
\be \label{eq:tanh2tau} 
\tanh(2\t_p) = - \frac{8 \pi \aa N^\k}{p^2 +8 \pi \aa N^\k}\,.
\ee
With these coefficients, we define the Bogoliubov transformation 
\be  \label{eq:defT}
T_\nu \; =  \exp \bigg(\;\frac 12 \sum_{p \in \L^*_+} \nu_p \big(a^*_p a^*_{-p} - \hc \big) \;\bigg)\,.
\ee
 For any $p \neq 0$ we have 
\be \label{eq:actionT}
T^*_\nu a_p T_\nu = \g_p a_p + \s_p a^*_{-p}
\ee
with the notation $\g_p = \cosh (\nu_p)$ and $\s_p = \sinh(\nu_p)$.  

With the Weyl operator (\ref{eq:defW}) and the Bogoliubov transformation (\ref{eq:defT}), we obtain the ``squeezed'' coherent state $\wt{\Psi}_N = W_{N_0} T_\nu \Omega$. Choosing $N_0$ so that $N = N_0 + \| \s_L \|^2$, one can show that this trial state has approximately $N$ particles and, to leading order, the correct ground state energy. However, as observed in \cite{ESY} (for a similar trial state) and later in \cite{NRS}, the energy of the quasi-free state $\wt{\Psi}_N$ does not match the second order correction in (\ref{eq:LHY}). To prove Prop. \ref{prop:energy-trial}, we need therefore to modify the trial state. We do so by replacing the vacuum $\Omega$ in the definition of $\wt{\Psi}_N$ by the normalized Fock space vector $\xi_\nu / \| \xi_\nu \|$, with $\xi_\nu = e^{A_\nu} \Omega$ and the cubic phase 
\be \label{eq:Adef}\begin{aligned}
A_\nu  =& \frac 1 {\sqrt {N}} \sum_{\substack{r\in P_H,v \in P_S:\\r+v\in P_H}} \eta_r \s_v \,a^*_{r+v} a^*_{-r} a^*_{-v} \Theta_{r,v}\,.
\end{aligned}\ee
Here, we introduced the momentum sets 
\begin{equation}\label{eq:PHS}  
\begin{split} 
P_H &= \{ p \in \L_+^* : |p| > N^{1-\kappa-\eps} \}\,, \\ 
P_S &=\Big\{p\in \L^*_+:  N^{\k/2-\eps} \leq |p|\leq N^{\k/2 + \eps}\Big\}\,.
\end{split} 
\end{equation}
Notice that $P_S \subset P_L$.  On the other hand, to make sure that $P_H \cap P_L = \emptyset$, we will require, from now on, that $\eps > 0$ is so small that $3\k - 2 + 4 \eps < 0$.

We denote by $\eta_L, \eta_{L^c}, \eta_S, \eta_H$ the restriction of $\eta : \L^* \to \bR$ to the set  $P_L, P_L^c, P_S$ and, respectively, $P_H$. Similarly, we define $\g_L, \g_{L^c}, \g_H, \g_S$ and $\s_L, \s_{L^c}, \s_H, \s_S$. The following lemma collects important bounds for these functions.
\begin{lemma} \label{lm:nu-norms}
We have  
\begin{align*}
&\| \eta_{L^c} \|^2 \leq C  N^{3 \k/2 -\eps} \,, \quad && \| \eta_{L^c} \|^2_{H^1} \leq C N^{1+\k}\,, \quad && \|\eta_{L^c} \|_\io \leq C N^{-2\eps} \\[0.2cm]
&\| \eta_H \|^2 \leq C  N^{3\k-1 +\eps} \,, \quad && \| \eta_H \|^2_{H^1} \leq C N^{1+\k} \,, \quad && \|\eta_H \|_\io \leq C N^{3\k-2 +2\eps} \,. 
\end{align*}
In particular, this implies that  $\|\g_H\|_\io, \|\s_H\|_\io \leq C$. Moreover we have  
\[
 \|\g_L \|^2_\io , \|\s_L \|^2_\io  \leq C N^{\k/2}  \,, \qquad \quad  \| \g_L \s_L \|_1 \leq C N^{3\k/2 +\eps}
\]
 and 
\[
 \| \g_L \|^2 \leq C N^{3\k/2 + 3\e}, \qquad  \| \s_L \|^2  \leq C N^{3\k/2} \, , \qquad \quad  \| \s_L \|^2_{H^1} \leq C N^{5\k/2 +\eps}\,.
\]
Finally, we observe that 
\begin{equation*}\label{eq:sgS} 
\| \s_S\|^2 \leq C N^{3\k/2}\,, \qquad  \| \s_S\|^2_{H^1} \leq C N^{5\k/2+\eps}\,,  \qquad \| \gamma_S \|_\infty^2, \| \sigma_S \|^2_\infty \leq C N^\eps \,.
\end{equation*}
\end{lemma}

\begin{proof} The bounds for $\|\eta_{L^c}\|$,  $\| \eta_H\|$, $\|\eta_{L^c}\|_\io$ and $\| \eta_H\|_\io$ follow from \eqref{eq:modetap}. 
On the other hand, with the notation $\check{\eta} (x) = - N w_\ell( N^{1-\k} x)$ for the function on $\L$ with Fourier coefficients $\eta_p$, we find from Lemma \ref{sceqlemma}, part iii), 
\[\label{eq:H1eta}
\| \eta_{L^c} \|_{H^1}^2 \leq C \sum_{p \in \L^*} p^2 |\eta_p|^2  \leq C \int |\nabla \check{\eta} (x)|^2 dx 
\leq C N^{1+ \kappa} \int_{\bR^3} \frac{1}{(|x|^2+1)^2}\, dx
\leq C N^{1+\k} \, .
\]
%
To show bounds for  $\s_L, \g_L$ we observe that, from \eqref{eq:tanh2tau}, 
\be \label{eq:sL2}
\s^2_p = \frac{p^2 + 8\pi \frak{a} N^\kappa - \sqrt{|p|^4 + 16 \pi  \frak{a} N^\kappa p^2}}{2 \sqrt{|p|^4 + 16 \pi \frak{a} N^\kappa p^2}} \,, \qquad \s_p \g_p = \frac{-8\pi \frak{a} N^\kappa}{2 \sqrt{|p|^4 + 16 \pi \frak{a} N^\kappa p^2}}\,.
\ee
Recalling that $P_L = \{ p \in \L^*_+ : |p| \leq  N^{\kappa/2+\eps} \}$, we find 
\begin{equation}\label{eq:sLinf} 
\| \s_L\|^2_\io  \leq \sup_{p \in P_L : |p| \leq N^{\kappa/2}} C \frac{N^{\kappa/2}}{|p|} + \sup_{p \in P_L : |p|\geq N^{\kappa/2}}  C \frac{N^{2\kappa}}{|p|^4}  \leq C N^{\kappa /2},
\end{equation} 
Moreover by definition of $P_S$ we also get
\[
	\|\s_S\|_\infty^2\leq \sup_{N^{\k/2-\eps}\leq|p|\leq N^{\k/2}} C \frac{N^{\k/2}}{|p|} + \sup_{p \in P_L: |p|\geq N^{\k/2}} C \frac{N^{2\k}}{|p|^4}  \leq C N^\e.
\]
With $\g^2_p= 1+ \s^2_p$, we find $\| \g_L\|_\io \leq C N^{\kappa/2}$ and $\|\g_S\|_\io\leq C N^\e$. Similarly, we obtain 
\be \label{eq:sL-2}
\| \s_L\|^2  \leq C  \sum_{p \in P_L : |p| \leq N^{\kappa/2}}  \frac{N^{\kappa/2}}{|p|} + C \sum_{p \in P_L : |p| > N^{\kappa/2}} \frac{N^{2\kappa}}{|p|^4} \leq C N^{3\kappa/2}
\ee
and thus, using again $\g^2_p= 1+ \s^2_p$, also $\|\g_L\|^2 \leq C N^{3\kappa/2 + 3 \eps}$.  Moreover, we have 
\be  \label{eq:sL-H1}
\| \s_L\|^2_{H_1} \leq  C \sum_{p \in P_L : |p| \leq N^{\kappa/2}} |p| N^{\kappa/2}   + C \sum_{p \in P_L : |p| \geq N^{\kappa/2}} \frac{N^{2\kappa} }{p^2}  \leq C N^{5\kappa/2 + \eps}.   
\ee
The bound  for  $\| \s_L \g_L\|_1$ is proved similarly, using the expression for $\g_p \s_p$ in \eqref{eq:sL2}. Finally, we note that the estimates \eqref{eq:sL-2} and \eqref{eq:sL-H1} do not improve when we consider the restriction of $\s_p$ to $P_S\subset P_L$, hence $\| \s_S\|^2  \leq C N^{3\k/2}$ and $\| \s_S\|^2_{H^1}  \leq C N^{5\k/2+\eps}$.
\end{proof}

In (\ref{eq:Adef}), we included, for every $r \in P_H$ and every $v \in P_S$, the cutoff 
\[\label{eq:defTheta} \begin{split}
\Theta_{r,v}= &\;  \!\!\prod_{s \in P_H}\!\Big[ 1-\chi(\cN_s>0)\chi(\cN_{-s+v}>0)\Big] \\ & \hspace{3cm} \times \prod_{w \in P_S}\! \Big[ 1-\chi(\cN_w>0)\chi(\cN_{r-w}+\cN_{-r-v-w}>0) \Big] 
\end{split}\]
where $\cN_p = a^*_p a_p$ and $\chi(t>0)$ is the characteristic function of the set $t>0$. 

\medskip

{\it Remark:} It is easy to check that the computation of the energy and the number of particles of the trial state we are constructing would not change substantially (and would still lead to a proof of Prop. \ref{prop:energy-trial}), if in the definition (\ref{eq:Adef}) of $A_\nu$ we restricted the sum over $r$ to the finite set $P_H \cap \{ p \in \L^*_+ : |p| < N^{1-\k + \eps} \}  = \{ p \in \L^*_+ : N^{1-\k -\eps} < |p| < N^{1-\k+\eps} \}$. With this choice, the infinite product over $s \in P_H$ appearing in the definition of the cutoff $\Theta_{r,v}$ would be replaced by a finite multiplication.  

\medskip

Let us briefly discuss the action of the cutoff $\Theta_{r,v}$. To understand its role in 
the computation of $e^{A_\nu} \Omega$, we observe that, for every integer $m \geq 2$, 
$r_1, \dots, r_m \in P_H$, $v_1, \dots , v_m \in P_S$, with $r_1 + v_1, \dots , r_m +v_m 
\in P_H$, we find 
\[ \begin{split} 
\Theta_{r_m,v_m} &a^*_{r_{m-1}+v_{m-1}} a^*_{-r_{m-1}}a^*_{-v_{m-1}}\dots a^*_{r_1+v_1}a^*_{-r_1}a^*_{-v_1}\O=\\
	&= \prod_{i,j=1}^{m-1}\prod_{\substack{p_\ell\in\{-r_\ell,r_\ell+v_\ell\},\\\ell=i,j,m}}\d_{p_i\neq -p_j+v_m}\d_{-p_m+v_i\neq p_j} \\ &\hspace{4cm} \times a^*_{r_{m-1}+v_{m-1}}a^*_{-r_{m-1}}a^*_{-v_{m-1}}\dots a^*_{r_1+v_1}a^*_{-r_1}a^*_{-v_1}\O.
\end{split}\] 
The choice $i=j$ in the product on the second line introduces restrictions of the form $v_m \neq v_i$ and $p_m \neq p_i$ where $p_\ell \in\{-r_\ell,r_\ell+v_\ell\}$ for $\ell = m,i$, for all $i \in \{1, \dots , m-1 \}$ (the condition $p_i \neq -p_i + v_m$, on the other hand, is trivially satisfied due to the assumption $p_i\in P_H,v_m\in P_S$). For $m\geq 3$, the cutoff $\Theta_{r_m,v_m}$ implements additional restrictions involving three indices of the form $-p_i+v_j \neq p_k$ with $p_\ell\in \{-r_\ell,r_\ell+v_\ell\},$ $\ell=i,j,k$ where $i,j,k = 1,\dots,m,$ $i\neq j\neq k$, so that exactly one of the three indices is $m$.
We conclude that, for any $m\geq 2$, 
\begin{equation}\label{eq:AmO}
\begin{split} 
A_\nu^m\O = &\frac1{N^{m/2}}\sum_{\substack{r_1\in P_H,v_1\in P_S:\\r_1+v_1\in P_H}}\dots\sum_{\substack{r_m\in P_H,v_m\in P_S:\\r_m+v_m\in P_H}}\prod_{i=1}^m \eta_{r_i}\s_{v_i}\\ &\hspace{2cm} \times \theta(\{r_j,v_j\}_{j=1}^m)a^*_{r_m+v_m}a^*_{-r_m}a^*_{-v_m}\dots a^*_{r_1+v_1}a^*_{-r_1}a^*_{v_1}\O
\end{split} 
\end{equation} 
where
\be \label{eq:theta}
\theta\big( \{r_j, v_j \}_{j=1}^{m} \big) = \prod_{\substack{i,j,k=1\\j\neq k}}^m\prod_{\substack{p_i\in\{-r_i,r_i+v_i\}\\p_k\in\{-r_k,r_k+v_k\}}}\d_{-p_i+v_j\neq p_k}.
\ee

To understand the reason for the introduction of the cutoff, let us compute the norm $\| \xi_\nu \|$ of the vector $\xi_\nu = e^{A_\nu} \Omega$. With (\ref{eq:AmO}), we find 
\begin{equation}\label{eq:normA0} \begin{split} 
 \| \xi_\nu\|^2 &= \sum_{m\geq 0}\frac{1}{(m!)^2} \| (A_\nu)^m \O \|^2 \\ 
 &=  \sum_{m\geq 0}\frac1{(m!)^2}\frac1{N^m} \sum_{\substack{v_1, \tl v_1 \in P_S \\ r_1, \tl r_1 \in P_H: \\ r_1 + v_1 ,\, \tl r_1 + \tl v_1 \in P_H}} \cdots 
\sum_{\substack{v_m, \tl v_m \in P_S \\ r_m, \tl r_m \in P_H: \\ r_m + v_m ,\, \tl r_m + \tl v_m \in P_H}} \theta\big(\{\tl r_j,\tl v_j\}_{j=1}^m\big) \theta\big(\{r_j,v_j\}_{j=1}^m\big)  \\
& \hspace{1cm}\times\prod_{i=1}^m \eta_{r_i}\eta_{\tl r_i} \s_{v_i} \s_{\tl v_i} \bmedia{a^*_{r_m+v_m}a^*_{-r_m}a^*_{-v_m}\dots a^*_{-v_1}\O, a^*_{\tilde{r}_m+\tilde{v}_m}a^*_{-\tilde{r}_m}a^*_{-\tilde{v}_m} \dots a^*_{-\tilde{v}_1}\O}\,.
\end{split} \end{equation}
Clearly, for the expectation on the last line not to vanish, all creation and annihilation operators with momenta in $P_S$ must be contracted among themselves. Since, on the support of $\theta (\{ r_j, v_j \}_{j=1}^m)$, $v_i \neq v_j$ for all $i\neq j$ (and similarly $\tl v_i \neq \tl v_j$ for all $i\neq j$ on the support of  $\theta (\{ \wt{r}_j, \wt{v}_j \}_{j=1}^m)$) we have $(m!)$ identical contributions arising from this pairing. We end up with 
\begin{multline} \label{eq:normexpA2}
	 \| \xi_\nu\|^2 
		=\sum_{m\geq 0}\frac1{m!}\frac1{N^m} \sum_{\substack{v_1 \in P_S\,, r_1, \tl r_1 \in P_H: \\ r_1 + v_1 ,\, \tl r_1 +  v_1 \in P_H}} \cdots 
\sum_{\substack{v_m \in P_S\,, r_m, \tl r_m \in P_H: \\ r_m + v_m ,\, \tl r_m + v_m \in P_H}}  \hskip -0.5cm \theta\big( \{r_j, v_j \}_{j=1}^{m} \big) \theta\big( \{ \tl{r}_j, v_j \}_{j=1}^{m} \big) \\ 
	\hspace{4.5cm}\times   \prod_{i=1}^m \eta_{r_i}\eta_{\tl r_i} \s^2_{v_i} \bmedia{\O,A_{r_1,v_1}\dots A_{r_m,v_m} A^*_{\tilde{r}_m,v_m} \dots A^*_{\tilde{r}_1,v_1} \O}
\end{multline}
where we have introduced the notation $A_{r_i,v_i} = a_{r_i+v_i} a_{-r_i}$.  

It is now important to observe that, because of the presence of the cutoffs, the annihilation operators in $A_{r_j, v_j}$ must be contracted with the creation operators in $A_{\wt{r}_j, v_j}$. In fact, if this was not the case, we would have $-r_j = -\wt{r}_\ell$ or $-r_j = \wt{r}_\ell + v_\ell$ and $r_j + v_j = -\wt{r}_k$ or $r_j + v_j = \wt{r}_k + v_k$, with at least one of the two indices $\ell, k$ different from $j$. This would imply one of the four relations $\wt{r}_\ell + v_j = -\wt{r}_k$, $\wt{r}_\ell + v_j = \wt{r}_k + v_k$, $ \wt{r}_\ell + v_\ell = \wt{r}_k + v_j$, $-\wt{r}_\ell - v_\ell + v_j = \wt{r}_k + v_k$, all of which are forbidden by the cutoff $\theta\big( \{ \wt{r}_j, \wt{v}_j \}_{j=1}^{m} \big)$. We conclude that 
\begin{multline} \label{eq:cAm}
 \bmedia{\O,A_{r_1,v_1}\dots A_{r_m,v_m} A^*_{\tilde{r}_m,v_m} \dots A^*_{\tilde{r}_1,v_1} \O} \theta\big( \{r_j, v_j \}_{j=1}^{m} \big) \theta\big( \{ \tl{r}_j, \tl{v}_j \}_{j=1}^{m} \big) \\
= \prod_{i=1}^m \Big(\d_{\tl r_i, r_i} + \d_{-\tl r_i, r_i+v_i}\Big) \theta\big( \{r_j, v_j \}_{j=1}^{m} \big)\,,
\end{multline}
(after identification of the momenta, the second cutoff becomes superfluous). From \eqref{eq:normexpA2}, we obtain 
\be \begin{split} \label{eq:normA}
\|\xi_\nu \|^2 &=  \sum_{m\geq 0}\frac1{m!}\frac1{N^{m}} \sum_{\substack{ v_1 \in P_S, r_1 \in P_H : \\  r_1 +v_1 \in P_H  } } \cdots \sum_{\substack{ v_m \in P_S, r_m \in P_H : \\  r_m +v_m \in P_H  } } \theta\big(\{r_j, v_j\}_{j=1}^{m}\big)  \, \prod_{i=1}^{m}\eta_{r_i} \big( \eta_{r_i} +\eta_{r_i+v_i} \big) \s_{v_i}^2\\ 
&=  \sum_{m\geq 0}\frac1{2^m m!}\frac1{N^{m}} \sum_{\substack{ v_1 \in P_S, r_1 \in P_H : \\  r_1 +v_1 \in P_H  } } \cdots \sum_{\substack{ v_m \in P_S, r_m \in P_H : \\  r_m +v_m \in P_H  } } \theta\big(\{r_j, v_j\}_{j=1}^{m}\big)  \, \prod_{i=1}^{m} \big( \eta_{r_i} +\eta_{r_i+v_i} \big)^2  \s_{v_i}^2\,.
\end{split}\ee
where we used the invariance of $\theta$, w.r.t. $-r_i \to r_i + v_i$. The cutoffs have been used first to exclude coinciding momenta in $v_1, \dots , v_m$ and in $\wt{v}_1, \dots , \wt{v}_m$ (which implies that, up to permutations, the pairing of the momenta in $P_S$ is unique) and then in (\ref{eq:cAm}) to make sure that annihilation operators in $A_{r_j, v_j}$ can only be contracted with the creation operators in $A^*_{\wt{r}_j, v_j}$. This substantially simplifies computations. Similar simplifications will arise in the computation of the energy of our trial state.  

A part from the formula (\ref{eq:normA}) for the norm $\| \xi_\nu \|^2$, we will also need bounds on  the expectation, in the state $\xi_\nu / \| \xi_\nu \|$, of the number of particles operator $\cN$, of $\cN^2$, of the kinetic energy operator $\cK$ and of the product $\cK \cN$. These bounds are collected in the next proposition, whose proof will be discussed in Sec. \ref{sec:cubicconj}. 
\begin{prop} \label{prop:Abounds}
Let $\xi_\nu=e^{A_\nu}\O$ with $e^{A_\nu}$ defined in \eqref{eq:Adef} with $\kappa \in (1/2 ;2/3)$ and $\eps > 0$ such that $3\k -2 + 4\eps < 0$. Then, under the assumptions of Theorem \ref{thm:main} we have 
\begin{align}
\frac{\langle \xi_\nu, \cN^{\,j} \xi_\nu \rangle}{\| \xi_\nu \|^2} & \leq  C N^{(9\k/2-2+\eps)j}  \label{eq:xiAxi}  
\end{align} 
and
\be \label{eq:cKcN}
\frac{ \bmedia{\xi_\nu,\cK \cN^{j-1} \xi_\nu}}{\|\xi_\nu\|^2} \leq C N^{5\k/2}N^{(9\k/2-2+\eps)(j-1)}
\ee
for $j=1,2$.
\end{prop}

Using the Weyl operator $W_{N_0}$ from (\ref{eq:defW}), the Bogoliubov transformation $T_\nu$ defined in (\ref{eq:defT}) and the cubic phase $A_\nu$ introduced in (\ref{eq:Adef}) (or, equivalently, the vector $\xi_\nu = e^{A_\nu} \Omega$), we can now define our trial state 
\be\label{eq:trial} \Psi_N = \frac{W_{N_0} T_\nu e^{A_\nu} \Omega}{\|  W_{N_0} T_\nu e^{A_\nu} \Omega \|} = W_{N_0} T_\nu \frac{\xi_\nu}{\| \xi_\nu \|} \,. \ee
Here, we choose $N_0 > 0$ such that 
\be \label{eq:fixN0}    
N = N_0 + \| \s_L \|^2
\ee
where we recall that $\s_L$ is the restriction to the set $P_L$ of the coefficients $\sigma_p = \sinh (\nu_p)$, with $\nu_p$ defining the Bogoliubov transformation $T_\nu$; see (\ref{eq:defT}). 

Let us briefly compare our trial state with the one of \cite{YY}. We both perturb the condensate with operators creating double and triple excitations, the latter having two particles with high momenta and one particle with low momentum. Moreover, we both impose cutoffs making sure that each low momentum appears only once. In contrast to \cite{YY}, here we also impose cutoffs on high momenta.  
Moreover, we have a clearer separation between creation of pairs (obtained through the Bogoliubov transformation $T_\nu$) and creation of triples. Finally, in our approach, we create triple excitations through the action of $e^{A_\nu}$ on the vacuum; the algebraic structure of the exponential makes the analysis and the combinatorics much simpler.  

As shown in the next proposition, the choice (\ref{eq:fixN0}) of $N_0$ guarantees that $\Psi_N$ has the right expected number of particles. 
\begin{prop} \label{prop:cN} Let $\Psi_N$ be defined in \eqref{eq:trial} with the parameter $N_0$ appearing in \eqref{eq:defW}  defined by \eqref{eq:fixN0}. Let $\k \in (1/2 ; 2/3)$ and $\e > 0$ so that $3\k -2 + 4\e < 0$. Then
\be \label{eq:<N>}
\langle \Psi_N, \cN \Psi_N \rangle \geq N \,,\qquad \langle \Psi_N, \cN^2 \Psi_N \rangle \leq C N^2
\ee
for all $N$ large enough. 
\end{prop}

\begin{proof} 
From \eqref{eq:actW} and \eqref{eq:actionT} we get
\be \begin{split} \label{eq:cM}
T_\nu^* W_{N_0}^* \cN W_{N_0} T_\nu  =\; &N_0+\sum_{p\in \L_+^*}\s_p^2+\sqrt{N_0}(a_0+a^*_0)+ a^*_0 a_0 \\
& +  \sum_{p\in \L^*_+}\big[ (\s_p^2+\g_p^2)a_p^*a_p + \g_p \s_p (a_p a_{-p} + \hc) \big] \,.
\end{split}\ee
Using that, by definition of $\xi_\nu$, $a_0 \xi_\nu = 0$ and $\langle \xi_\nu, a_p a_{-p} \xi_\nu \rangle = 0$ for every $p \in \Lambda^*$, as well as the definition $N = N_0 + \| \s_L \|^2$, we obtain that 
\begin{equation}\label{eq:numpar} \langle \Psi_N, \cN \Psi_N \rangle = \frac{\langle \xi_\nu , T_\nu^* W_{N_0}^* \cN W_{N_0} T_\nu \xi_\nu \rangle}{\| \xi_\nu \|^2} =  N +\sum_{p\in P_L^c}\s_p^2 +  \sum_{p\in P_S\cup P_H} (\s_p^2+\g_p^2)\frac{\langle \xi_\nu,a_p^*a_p \xi_\nu\rangle}{\|\xi_\nu\|^2} \,.
\end{equation} 
This immediately implies that $\langle \Psi_N, \cN \Psi_N \rangle \geq N$.

With $a_0 \xi_\nu = 0$ and the assumption $3\k - 2 + 4 \eps < 0$, (\ref{eq:cM}) also implies that 
\[ \begin{split} 
\langle \xi_\nu , T_\nu^* W_{N_0}^* \cN^2 W_{N_0} T_\nu \xi_\nu \rangle \leq  \; &C N^2 + C \sum_{p,q \in \L^*_+} (\s_p^2 + \g_p^2) (\s_q^2 + \g_q^2) \langle \xi_\nu, a_p^* a_p a_q^* a_q \xi_\nu \rangle \\ &+ C \sum_{p,q \in \L^*_+} \g_p \s_p \g_q \s_q \langle \xi_\nu, (a_p a_{-p} + a^*_p a_{-p}^*) (a_q a_{-q} + a_q^* a_{-q}^*) \xi_\nu \rangle\,. \end{split} \]
Since $\xi_\nu$ is a superposition of states with $3m$ particles with momenta in $P_H \cup P_S$, we obtain, writing $a_p a_{-p} a_q^* a_{-q}^* = a_q^* a_p a_{-q}^* a_{-p} +(\d_{p,q}+ \delta_{-p, q})( a_{p}^* a_{p} +1) + \delta_{p,q} a^*_{-p} a_{-p}$ and similarly for $a_p^* a_{-p}^* a_q a_{-q}$,  
\[ \begin{split} 
\langle \xi_\nu , T_\nu^* W_{N_0}^* \cN^2 W_{N_0} T_\nu \xi_\nu \rangle \leq \; &CN^2 + C N^{2\eps} \langle \xi_\nu , (\cN+1)^2 \xi_\nu \rangle  + C \sum_{p \in P_H \cup P_S} \g_p^2 \s_p^2 \langle \xi_\nu, a_p^* a_p \xi_\nu \rangle \\ & + C \sum_{p \in \L^*_+} \g_p^2 \s_p^2 \| \xi_\nu \|^2 + C \hspace{-.3cm}  \sum_{p,q \in P_H\cup P_S} \hspace{-.3cm} \g_p \s_p \g_q \s_q \langle \xi_\nu, a^*_p a_q a^*_{-p} a_{-q} \xi_\nu \rangle \,.
 \end{split} \]
Estimating the last term through  
\[ \begin{split} \Big|  \sum_{p,q \in P_H \cup P_S} \g_p  &\s_p \g_q \s_q \langle \xi_\nu, a^*_p a_q a^*_{-p}  a_{-q} \xi_\nu \rangle  \Big| \\ \leq \; & \sum_{p,q \in P_H \cup P_S} |\g_p | | \s_p | |\g_q | | \s_q | \,  \| a^* _q a_p \xi_\nu \| \| a^*_{-p}  a_{-q} \xi_\nu \| \\  \leq \; &  
\sum_{p,q \in P_H \cup P_S} | \g_p | |\s_p | |\g_q | |\s_q | \, \big[ \| a_q a_p \xi_\nu \| + \| a_p \xi_\nu \| \big] \big[ 
 \| a_{-p}  a_{-q} \xi_\nu \| + \| a_{-q} \xi_\nu \| \big]  \\
  \leq \; & C \| \g_{S \cup H} \|_\infty^2 \| \s_{S \cup H} \|_\infty^2 \| (\cN+1) \xi_\nu \|^2 + C \| \g_{S \cup H} \|_\infty^2 \| \s_{S \cup H} \|^2 \| \cN^{1/2}  \xi_\nu \|^2\\ &+  C \| \g_{S \cup H} \|_\infty^2 \| \s_{S \cup H} \|_\infty  \| \s_{S \cup H} \| \| (\cN+1) \xi_\nu \| \| \cN^{1/2} \xi_\nu \|  \end{split} \] 
we conclude with the bounds in Lemma \ref{lm:nu-norms} and in Prop. \ref{prop:Abounds} that, for $3\k -2 + 4\e < 0,$
\[ \begin{split} \label{eq:cM}
\langle \xi_\nu , T_\nu^* W_{N_0}^* \cN^2 W_{N_0} T_\nu \xi_\nu \rangle \leq \;& C N^2 + C N^{2\eps} \langle \xi_\nu, \cN^2 \xi_\nu \rangle + C N^\eps \| \sigma \|^2 \langle \xi_\nu, \cN \xi_\nu \rangle + C \| \g \|_\infty^2 \| \s \|^2 \\ \leq \; & C N^2 + C N^{9\k -4 + 4\e} + C N^{6\k -2 +2\eps} + C N^{2\kappa} \leq C N^2\,.
\\[-0.5cm]\qedhere \end{split} \]
\end{proof}

The next theorem, whose proof occupies the rest of the paper, determines the energy of $\Psi_N$ and, combined with Prop. \ref{prop:cN}, allows us to conclude the proof of Prop. \ref{prop:energy-trial}. \begin{theorem} \label{thm:CN}
Let $\cH_N$ and $\Psi_N \in \cF$ be defined as in \eqref{eq:HN-Fock} and \eqref{eq:trial} respectively, and let $E_N^\Psi= \langle \Psi_N, \cH_N \Psi_N \rangle$. Let $\k \in (1/2 ; 2/3)$, $\eps > 0$ so small that $3\k -2 + 4\e < 0$. Then, under the assumption of  Theorem \ref{thm:main}, we have 
\begin{equation}\label{eq:thmCN}
E_N^\Psi  \leq   4 \pi \aa N^{1+\k}  \bigg(1+\frac{128}{15\sqrt{\pi}}(\aa^3N^{3\k-2})^{1/2}\bigg)+  C N^{5\k/2}   \max \{ N^{-\e} \hspace{-.1cm} , N^{9\k-5 +6\eps}  \hspace{-.1cm} ,N^{21\k/4-3+3\e}\} 
\end{equation} 
for all  $N$ large enough. 
\end{theorem}

{\it Remark:} Eq. (\ref{eq:thmCN}) gives the correct second order term for all $\kappa < 5/9$ (choosing $\eps > 0$ small enough); this corresponds to exponents $\gamma < 4/3$ in Prop. \ref{prop:energy-trial}. With a more complicated proof, we could have considered all $\kappa < 7/12$ (corresponding to $\gamma < 5/3$). . 

\begin{proof}[Proof of Prop. \ref{prop:energy-trial}] 
Prop. \ref{prop:energy-trial} follows from Prop.\ref{prop:cN} and Theorem \ref{thm:CN}, recalling that (\ref{eq:cH}) is unitarily equivalent to $L^{-2} \cH_N$, with $\cH_N$ as defined in (\ref{eq:HN-Fock}), $L = \tl{\rho}^{-\g}$, $N = \tl{\rho} L^3 = \tl{\rho}^{1-3\g}$ and $\kappa = (2\g - 1)/ (3\g-1)$. At the end, to obtain (\ref{eq:tlrho-def}) and (\ref{eq:enL}), we have to rename $\eps (3\g - 1) \to \eps$.
\end{proof}

\section{Energy of the trial state} \label{sec:energy}

In this section we prove Thm. \ref{thm:CN}. With \eqref{eq:trial} and introducing the notations
\be \label{eq:defcG}
\cG_N = T^*_\nu \cL_N T_\nu\,, \quad \text{with} \quad \cL_N = W^*_{N_0} \cH_N W_{N_0}\,.
\ee
 we write
\[ \label{eq:E-PsiN}
E^{\Psi}_N = \langle \Psi_N, \cH_N \Psi_N \rangle = \frac{\langle \xi_\nu , \cG_N \xi_\nu \rangle}{\| \xi_\nu \|^2 }
\]
with $\xi_\nu = e^{A_\nu} \O$, as defined before (\ref{eq:Adef}). With \eqref{eq:HN-Fock}, and recalling from (\ref{eq:actW}) that  
\be \label{eq:actionW}
W^*_{N_0} a_p W_{N_0}=a_p+\sqrt{N_0}\,\d_{p,0}
\ee
we obtain
$\cL_N= \cL^{(0)}_N +\cL^{(1)}_N+ \cL^{(2)}_N + \cL^{(3)}_N + \cL^{(4)}_N $, with
\be \begin{split}\label{eq:cL}
	\cL^{(0)}_N=\;&\frac {N_0^2} {2}  \,N^{\k-1} \widehat{V}(0)\\
	\cL^{(1)}_N=\; & N_0^{3/2}\, N^{\k-1}\widehat{V}(0)(a_0+\hc)\\
	\cL^{(2)}_N=\;&\sum_{p\in \L^*}p^2a^*_pa_p + \frac{N_0}{N} \sum_{p\in \L^*} N^\k \big( \widehat{V}(p/N^{1-\k}) + \widehat V(0) \big) a_p^*a_p \\
& \hskip 3cm + \frac{N_0}{2 N}\sum_{p\in \L^*}N^\k \widehat{V}(p/N^{1-\k})(a_p^*a_{-p}^*+\hc)\\
	\cL^{(3)}_N=\; & \frac{\sqrt{N_0}}N \sum_{p,r\in \L^*} N^{\k} \widehat{V}(r/N^{1-\k})(a^*_pa^*_ra_{p+r}+\hc)\\
	\cL^{(4)}_N=\; &\frac1{2N}\sum_{p,q,r \in \L^*} N^{\k} \widehat{V}(r/N^{1-\k})a^*_{p+r} a^*_q a_{p}a_{q+r}\,.
\end{split}\ee

To compute $\cG_N$, we have to conjugate the operators in (\ref{eq:cL}) with the Bogoliubov transformation $T_\nu$. The result is described in the following proposition, whose proof will be discussed in Sec. \ref{sec:Bog}. 
\begin{prop}  \label{prop:calG}
Let 
\be  \label{eq:cHN}
\cK = \sum_{p \in \L^*_+} p^2 a^*_p a_p\,,  \qquad
\cV_N^{(H)}= \frac{1}{2 N}\sum_{\substack{r\in \L^*,\, p,q \in P_H:\\ p+r,q+r\in P_H}}N^\k \widehat{V}(r/N^{1-\k}) a_{p+r}^*a_q^*a_{p}a_{q+r} \,.
\ee
and (recalling from (\ref{eq:fixN0}) that $N_0 = N - \| \s_L \|^2$) 
\be\label{eq:cCG} 
\cC_N =\; \frac{\sqrt{N_0}}{ N} \sum_{\substack{p,r \in P_H\\ p+r \in P_S}} N^\k \widehat V(r/N^{1-\k})\, \sigma_{p+r} \g_p\g_r  \,(a^*_{p+r} a^*_{-p} a^*_{-r}  + \hc )\,. 
\ee
Moreover, let
\be \begin{split} \label{eq:constantG}
 C_{\cG_N} =\; &  \frac{N^{1+\k}}{2}  \widehat{V}(0) +\sum_{p\in {\L}^*_+}p^2\s_p^2 + \sum_{p\in \L^*_+} N^\k \widehat{V}(p/N^{1-\k})\s_p\g_p\\
&+ \sum_{p\in P_L}N^\k \widehat{V}(p/N^{1-\k}) \s_p^2   +\frac{1}{2N}\sum_{\substack{p,r\in \L^*_+\\r\neq p}} N^\k \widehat{V}(r/N^{1-\k})\s_p\s_{p-r}\g_p\g_{p-r} \\
&  -  \frac 1 N \sum_{v \in P_L} \s^2_v \sum_{p \in P_L^c}  N^\k \widehat V(p/N^{1-\k}) \eta_p \,.
\end{split}\ee
Then, under the assumptions of Theorem \ref{thm:main},  for all $\k \in (1/2 ; 2/3)$ and $\eps > 0$ with $3\k -2 + 4\eps < 0$ and $N$ sufficiently large, we have
\begin{equation}\label{eq:en-G}
\begin{split} \frac{\langle \xi_\nu, \cG_N \xi_\nu \rangle}{\| \xi_\nu \|^2} \leq \; &C_{\cG_N} + \frac{\langle \xi_\nu, (\cK+\cV_N^{(H)}  +\cC_N) \xi_\nu\rangle}{\| \xi_\nu \|^2}  \\ &+ C N^{5\k/2} \cdot  \max \{ N^{-\eps}, N^{9\k - 5 + 6\e} , N^{21 \kappa /4 - 3 + 3 \eps } \}\,. \end{split} 
\end{equation} 
\end{prop}

The expectation of the operators $\cK, \cV_N^{(H)}, \cC_N$ in the state $\xi_\nu / \| \xi_\nu \|$, appearing  on the r.h.s. of (\ref{eq:en-G}) is determined by the next proposition, which will be shown in Sec. \ref{sec:cubicconj}. 
\begin{prop} \label{prop:cubic-contr}  
Under the assumptions of Theorem \ref{thm:main}, we have  
\[\label{eq:prop:cubic-contr}
\begin{split} 
\frac{\langle \xi_\nu, (\cK+\cV^{(H)}_N  +\cC_{N}) \xi_\nu \rangle }{\|\xi_\nu\|^2} \leq \; & \frac1N\sum_{v\in P_L}\s_v^2\sum_{r\in P_H} N^\k \widehat V(r/N^{1-\k})(\eta_r+\eta_{r+v})  \\ &+  C N^{5\k/2}  \cdot  \max \{ N^{-\e} ,\, N^{12\k-7 +5\e}\} 
\end{split} 
\]
for all $\k \in (1/2 ; 2/3)$ and all $\e > 0$ so small that $3\k -2 + 4\e < 0$. 
\end{prop}

Let us now use the statements of Prop. \ref{prop:calG} and Prop. \ref{prop:cubic-contr} to obtain an upper bound for the energy of the trial state $\Psi_N$ and prove Thm. \ref{thm:CN}. From Prop. \ref{prop:calG} and Prop.  \ref{prop:cubic-contr} we find 
\be \begin{split} \label{eq:const1}
E_N^\Psi \leq \; &  \frac{N^{1+\k}}{2}\,\widehat{V}(0)+ \sum_{p \in \L^*_+} p^2 \s_p^2 + \sum_{p \in \L^*_+}  N^\k \widehat V(p/N^{1-\k}) \g_p \s_p      \\
& + \sum_{v \in P_L} N^\k \widehat V(v/N^{1-\k}) \s^2_v   + \frac 1 {2N} \sum_{\substack{p,q \in \L^*_+\\ p \neq q}} N^\k \widehat V( (p-q)/N^{1-\k}) \g_p \g_{q} \s_p \s_{q}\\
& -  \frac 1 N \sum_{v \in P_L} \s^2_v \sum_{p \in P_L^c}  N^\k \widehat V(p/N^{1-\k}) \eta_p 
\\ &
+ \frac1N\sum_{v\in P_L}\s_v^2\sum_{r\in P_H} N^\k \widehat V(r/N^{1-\k})(\eta_r+\eta_{r+v})\, + \cE 
\end{split}\ee
with 
\[  \label{eq:EPsiN-CN}
\cE \leq C N^{5\k/2} \cdot  \max \{ N^{-\e},\,N^{9\k-5 +6\eps} ,\, N^{21\k/4-3+3\e}\}
\]
for all $\k \in (1/2 ; 2/3)$ and all $\e > 0$ with $3\k -2 + 4\e < 0$ (in this range, $12\k-7+5\e < 9\k -5+3\e$). Since $|\sigma_p - \eta_p| \leq C |\eta_p|^3$ for all $p \in P_L^c$, with (\ref{eq:modetap}) we find
\be \label{eq:const1-s1}
\sum_{p \in \L^*_+} p^2 \s_p^2 \leq \sum_{v \in P_L} v^2 \s_v^2  + \sum_{p \in P_L^c} p^2 \eta_p^2 + C N^{5\kappa/2 - 3\eps} \,. 
\ee
Similarly, with $|\g_p \s_p - \eta_p| \leq C \eta_p^3$ for all $p  \in P_L^c $,  the last term on the first line of \eqref{eq:const1} can be written as
\be \label{eq:const1-s2}
\begin{split} 
\sum_{p \in \L^*_+}   N^\k &\widehat V(p/N^{1-\k})\g_p \s_p  \\ & \leq \sum_{v \in P_L} N^\k  \widehat V(v/N^{1-\k}) \g_v \s_v + \sum_{p \in P_L^c} N^\k \widehat V(p/N^{1-\k})  \eta_p + C N^{5\kappa/2-3\eps} \,.\end{split} 
\ee
Next, we focus on the last term on the second line of \eqref{eq:const1}. We write 
\be \begin{split} \label{eq:const1-s3}
&\frac 1 {2N} \sum_{\substack{p,q \in \L^*_+\\ p \neq q}} N^\k \widehat V((p-q)/N^{1-\k}) \g_p \g_{q} \s_p \s_{q} \\ &= \frac 1 {N} \!\!\sum_{\substack{p \in P_L^c, \\ q \in P_L}} N^\k \widehat V((p-q)/N^{1-\k})  \eta_p \g_{q} \s_{q}  +  \frac 1 {2N}\!\! \sum_{\substack{p,q \in P_L^c\\ p \neq q}} \!N^\k \widehat V((p-q)/N^{1-\k})   \eta_p \eta_q \!+ \cE_1\,.
\end{split}\ee 
Using again $|\gamma_p \sigma_p - \eta_p | \leq C |\eta_p|^3$ for all $p \in P_L^c$, the estimate 
\be  \label{eq:intVeta}
\sup_{r\in \L^*}\,\sum_{s \in \L^*_+}  N^\k |\widehat V((r-s)/N^{1-\k})| |\eta_s| \leq C N^{1+\k}  
\ee
and the bounds from Lemma \ref{lm:nu-norms}, we conclude (using the condition $3\k -2 + 5\eps < 0$) that 
\[ \begin{split} \cE_1 &\leq C \left[  N^{\kappa} \| \eta_{L^c} \|_\infty \| \eta_{L^c} \|^2 + N^{\kappa-1} \| \eta_{L^c} \|_\infty \| \eta_{L^c} \|^2 \| \sigma_L \gamma_L \|_1 + N^{\kappa-1} \| \sigma_L \gamma_L \|_1^2  \right] \leq C N^{5\kappa/2-\eps}\,. \end{split} \]
To prove (\ref{eq:intVeta}) we use (\ref{eq:modetap}) and we remark that   
\[ \sum_{s \in \L^* : |s| \leq N^{1-\k}} N^\k   |\widehat V((r-s)/N^{1-\k})| |\eta_s| \leq C N^{2\k} \sum_{s \in \L^*_+: |s| \leq N^{1-\k}} |s|^{-2} \leq C N^{1+\k} \]
and that, rescaling variables (setting $\tl{r} = r/ N^{1-\k}$) and using an integral approximation, 
\[ \begin{split} \sum_{s \in \L^*_+ : |s| > N^{1-\k}} N^\k   |\widehat V((r-s)/N^{1-\k})| |\eta_s| &\leq C N^{1+\k} \sum_{s \in \L^* / N^{1-\k}: |s| > 1} N^{-3(1-\k)} \, |\widehat{V} (\tl{r} - s)| |s|^{-2} \\ &\leq C N^{1+\k} \int_{|s| > 1}  |\widehat{V} (\tl{r} - s)| |s|^{-2} ds \leq C_q N^{1+\k} \| \widehat{V} \|_q  \end{split} \]
for any $q < 3$. With the assumption $V \in L^{q'} (\bR^3)$, for some $q' > 3/2$, (\ref{eq:intVeta}) follows by the Hausdorff-Young inequality. 

Finally, we remark that the terms in the last two lines of \eqref{eq:const1} can be combined, using (\ref{eq:modetap}) and the bound $\| \sigma_L  \|^2 \leq C N^{3\kappa/2}$ from Lemma \ref{lm:nu-norms} into 
\begin{multline} \label{eq:const1-s4}
 -  \frac 1 N \sum_{v \in P_L} \s^2_v  \sum_{p \in  P_L^c} N^\k \widehat V(p/N^{1-\k}) \eta_p   +\frac1N\sum_{v\in P_L}\s_v^2\sum_{r\in P_H} N^\k \widehat V(r/N^{1-\k})(\eta_r+\eta_{r+v})\\
\leq \frac 1 N \sum_{v \in P_L} \s^2_v \sum_{r \in P_L^c} N^\k \widehat V(r/N^{1-\k}) \eta_{r+v} + C N^{5\kappa/2 - \eps}\, . 
\end{multline}

Inserting \eqref{eq:const1-s1}, \eqref{eq:const1-s2}, \eqref{eq:const1-s3} and \eqref{eq:const1-s4}  into \eqref{eq:const1}, we obtain 
\be \begin{split} \label{eq:const2}
E_N^\Psi  \leq & \,  \frac{N^{1+\k}}{2}  \widehat{V} (0)+ \sum_{p \in  P_L^c}\big[ p^2 \eta_p + N^\k \widehat V(p/N^{1-\k}) +   \frac 1 {2N} \sum_{r \in  P_L^c}  N^\k \widehat V((r-p)/N^\k)  \eta_r  \big] \eta_p  \\
& + \sum_{v \in P_L} \Big[ v^2 \s_v^2 +  \big(\s^2_v + \g_v \s_v \big)\Big( N^\k \widehat V(v/N^{1-\k}) +\frac 1 N  \sum_{r \in  P_L^c}  N^\k \widehat V((r-v)/N^{1-\k})\eta_{r} \Big) \Big]  \\
& +  C N^{5\k/2}   \max \{ N^{-\e},\,N^{9\k-5 +6\eps} ,\, N^{21\k/4-3+3\e}\} \, . 
\end{split}\ee

Let us now consider the first square bracket on the r.h.s. of \eqref{eq:const2}. Using the scattering equation \eqref{eq:scatt} we obtain
\begin{multline} \label{eq:const2-step1}
\sum_{p \in  P_L^c} \Big[ p^2 \eta_p +  N^\k \widehat V(p/N^{1-\k})  +   \frac 1 {2N} \sum_{r \in  P_L^c}  N^\k \widehat V( (r-p)/N^{1-\k})  \eta_r  \Big] \eta_p  \\
= \frac 1 2 \sum_{p \in  P_L^c}  N^\k \widehat V(p/N^{1-\k}) \eta_p - \frac 1 {2N} \sum_{\substack{p \in  P_L^c\\ v \in P_L} } N^\k \widehat V((p-v)/N^{1-\k})  \eta_v  \eta_p+ \cE_2
\end{multline}
with 
\[ 
\cE_2 =\; N^{3-2\k} \l_{\ell} \sum_{p \in  P_L^c}  \big( \widehat \chi_\ell \ast \widehat f_{N}\big)_p \eta_p  - \frac 1 {2N} \sum_{p \in  P_L^c}  N^\k \widehat V(p/N^{1-\k})\eta_p  \eta_0 \,.
\]
Using $N^{3-3\k} \l_{\ell} \leq C$ (Lemma \ref{sceqlemma}), $\| \widehat \chi_\ell \ast \widehat f_{N}\|  = \| \chi_\ell f_N \| \leq C$ and $\|\eta_{L^c}\|^2 \leq C N^{3\k/2 -\e}$ (Lemma \ref{lm:nu-norms}) in the first term, (\ref{eq:wteta0}) and \eqref{eq:intVeta} in the second term, we find $\cE_2 \leq C N^{5\kappa/2-\eps}$ (using that $3\kappa -2 + 4\eps < 0$ and $\kappa > 1/2$). 

As for the second square bracket on the r.h.s. of \eqref{eq:const2}, we write
\begin{multline} \label{eq:const2-step2}
\sum_{v \in P_L}  \big(\s^2_v + \g_v \s_v \big) \Big( N^\k \widehat V(v/N^{1-\k})  +\frac 1 N \sum_{r \in  P_L^c} N^\k \widehat V((r-v)/N^{1-\k})\eta_{r} \Big) \\
=  \sum_{v \in P_L}  \big(\s^2_v + \g_v \s_v \big)  N^\k \big( \widehat V(\cdot/N^{1-\k}) \ast \widehat f_{N} \big)_v + \cE_3
\end{multline}
with 
\[ 
\cE_3 =\;  -\frac 1 N \hskip -0.2cm \sum_{v,\,r\in P_L}\hskip -0.2cm N^\k \widehat V((r-v)/N^{1-\k})  \big(\s^2_v + \g_v \s_v \big)\eta_{r} - \frac {\eta_0} N  \sum_{v \in P_L} N^\k \widehat V_{N}(v/N^{1-\kappa}) \big(\s^2_v + \g_v \s_v \big) \,.
\]
With (\ref{eq:modetap}), Lemma \ref{lm:nu-norms}, $|\eta_0| \leq C N^\k$ and the assumption $3\k -2 + 5\eps < 0$, we obtain $\cE_3 \leq N^{5\k/2-\eps}$. 

Inserting \eqref{eq:const2-step1} and \eqref{eq:const2-step2} in \eqref{eq:const2} and completing sums over $p$ on the r.h.s. of (\ref{eq:const2-step1}), we arrive at 
\begin{equation} \label{eq:const4}
\begin{split} 
E_N^\Psi \leq \; & \frac{N}{2} \big( N^\k \widehat V_{N}(\cdot/N^{1-\k}) \ast \widehat f_{N} \big)_0    \\
&+ \sum_{v \in P_L} \Big[ v^2 \s_v^2 +  \big(\s^2_v + \g_v \s_v \big)  \big( N^\k \widehat V(\cdot/N^{1-\k}) \ast \widehat f_{N} \big)_v - \frac {1} 2    N^\k \big( \widehat V(\cdot/N^{1-\k}) \ast \widehat f_{N} \big)_v  \eta_v   \Big]  \\ &+ C N^{5\k/2}   \max \{ N^{-\e},\,N^{9\k-5 +6\eps} ,\, N^{21\k/4-3+3\e}\}\,.
\end{split} 
\end{equation} 

Let us now introduce the notation $\hat{g}_p = ( N^\k \widehat V(\cdot/N^{1-\k}) \ast \widehat f_{N})_p$. Notice that 
\be  \label{eq:Vastf}
 \big|\, \widehat g_0 - 8 \pi \aa N^\k \,\big| \leq C N^{2\k -1}, \qquad
 \big|\, \widehat g_p  - \widehat g_0  \, \big|   \leq C |p| N^{2\k -1} \,.
\ee
With the expression \eqref{eq:sL2}, we obtain 
\[ \begin{split}
 \sum_{v \in P_L} \Big[ v^2 \s_v^2 +  \big(\s^2_v + \g_v \s_v \big)  \widehat g_v \Big] &  =\frac 12 \sum_{v \in P_L} \left[ -v^2 - \widehat g_v  + \frac{v^4 + v^2 (8\pi \frak{a} N^\kappa + \widehat g_v)}{\sqrt{v^4 + 16 \pi \frak{a} N^\kappa v^2}} \right] \\
& = \frac 12 \sum_{v \in P_L} \left[ \sqrt{v^4 + 16 \pi \frak{a} N^\k v^2}  - v^2 - 8 \pi \aa N^\k \right] + \cE_{4}
\end{split} 
\]
where, with \eqref{eq:Vastf} and $3\k -2 + 4\eps < 0$, we find 
\[ \begin{split} 
\cE_{4} &= \frac 12 \sum_{v\in P_L} \left[ 8\pi \frak{a} N^\kappa - \widehat g_v + \frac{\widehat g_v - 8\pi \frak{a} N^\kappa}{\sqrt{1 + 16 \pi \frak{a} N^\kappa / v^2}} \right] \\ &\leq C \sum_{v\in P_L}  | 8\pi \frak{a} N^\kappa - \widehat g_v | \;\frac{|\sqrt{1 + 16 \pi \frak{a} N^\kappa / v^2} - 1 |}{\sqrt{1 + 16 \pi \frak{a} N^\kappa /v^2}}  \\ &\leq  C N^{2\k-1} \Big[ \sum_{|v| < N^{\kappa/2}} (1+ |v|) + \sum_{v \in P_L : |v| > N^{\kappa/2}} \frac{N^{\kappa/2}}{|v|^2} \Big] \leq C N^{5\k/2-\eps} \, . \end{split} 
\]
Moreover, from (\ref{eq:const4}) and with the scattering equation \eqref{eq:scatt}, we obtain 
\[  - \frac{N^\kappa}{2} \sum_{v \in P_L} \widehat{g}_v \eta_v \leq \sum_{v \in P_L} \frac{(8\pi \frak{a} N^\kappa )^2}{4v^2} + C N^{5\kappa/2-\eps}\,. \] 
Thus 
\[ \begin{split} E^\Psi_N \leq \; &  4 \pi \aa N^{1+\k} + \frac 12 \sum_{v \in P_L} \Big[ \sqrt{v^4 + 16 \pi \aa N^\k v^2} - v^2 - 8 \pi \aa N^\k + \frac{(8 \pi \aa N^\k )^2}{2v^2}   \Big]  \\ &+  C N^{5\k/2}   \max \{ N^{-\e},\,N^{9\k-5 +6\eps} ,\, N^{21\k/4-3+3\e}\} \,.
\end{split}\]
With 
\[ \left|  \sqrt{v^4 + 16 \pi \aa N^\k v^2} - v^2 - 8 \pi \aa N^\k + \frac{(8 \pi \aa N^\k )^2}{2v^2}  \right| \leq C \frac{N^{2\kappa}}{|v|^4} \]
we can replace, up to an error of order $N^{5\k/2-\e}$, the sum over $P_L$ with a sum over all $\Lambda^*_+$. With the rescaling $v \to N^{\kappa/2} v$, we arrive at 
\be \begin{split} \label{eq:const5}
E^\Psi_N \leq \; &  4 \pi \aa N^{1+\k} + \frac{N^\kappa}{2} \sum_{v \in 2\pi N^{-\kappa/2} \bZ^3} \Big[ \sqrt{v^4 + 16 \pi \aa v^2} - v^2 - 8 \pi \aa + \frac{(8 \pi \aa )^2}{2v^2}   \Big]  \\ &+  C N^{5\k/2}   \max \{ N^{-\e},\,N^{9\k-5 +6\eps} ,\, N^{21\k/4-3+3\e}\}\,.
\end{split} \ee 
Recognizing that (\ref{eq:const5}) defines a Riemann sum and explicitly computing 
\[ \frac{1}{2 (2\pi)^3} \int dq \Big[ \sqrt{v^4 + 16 \pi \aa v^2} - v^2 - 8 \pi \aa + \frac{(8 \pi \aa )^2}{2v^2}   \Big] = 4 \pi \frak{a} \cdot \frac{128}{15 \sqrt{\pi}} \frak{a}^{3/2} \] 
we conclude that 
\[ \begin{split} 
E^\Psi_N \leq \; &  4 \pi \aa N^{1+\k}  \cdot \left[ 1 + \frac{128}{15 \sqrt{\pi}} \, (\frak{a}^3 N^{3\kappa-2})^{1/2} \right] \\ &+  C N^{5\k/2}   \max \{ N^{-\e},\,N^{9\k-5 +6\eps} ,\, N^{21\k/4-3+3\e}\} \,.
\end{split} \]
To compare the Riemann sum in (\ref{eq:const5}) with the integral, we first removed contributions arising from $|v| \leq N^{-\eps}$ using that $|F(v)| \leq C / v^2$, for small $v$, with the definition $F(v) = \sqrt{v^4 + 16 \pi \aa v^2} - v^2 - 8 \pi \aa + (8 \pi \aa )^2 / 2v^2$. For $|v| > N^{-\eps}$, we use that $|\nabla F (v)| \leq C |v|^{-3} (1+v^2)^{-1}$ to compare the value of $F(q)$ with $F(v)$, for all $q$ in the cube of size $2\pi N^{-\k/2}$ centered at $v$.


\section{Bogoliubov transformation} \label{sec:Bog}

In this section, we show Prop. \ref{prop:calG}. From the definition (\ref{eq:defcG}) and from (\ref{eq:cL}), we obtain (since $T_\nu$ does not act on the zero momentum mode and since $a_0 \xi_\nu = 0$) 
\begin{equation*} \label{eq:cGN-deco} \frac{\langle \xi_\nu, \cG_N \xi_\nu \rangle}{\| \xi_\nu \|^2}  = \frac{N_0^2}{2} N^{\kappa - 1} \widehat{V} (0) + \sum_{j=2}^4 \frac{\langle \xi_\nu, \cG_N^{(j)} \xi_\nu \rangle}{\| \xi_\nu \|^2} \end{equation*} 
with $\cG_N^{(j)} = T_\nu^* \cL_N^{(j)} T_\nu$, for $j=2,3,4$. 

We start from the contribution of $\cG_N^{(2)}$. We write $\cL_N^{(2)}= \cK + \cL_N^{(2,V)}$ with 
\[
\cL_N^{(2,V)} =\frac{N_0}{N}\sum_{p\in \L^*} N^\k \big( \hat{V}(p/N^{1-\k}) + \hat V(0) \big) a_p^*a_p + \frac{N_0}{2N}\sum_{p\in \L^*}N^\k \hat{V}(p/N^{1-\k})(a_p^*a_{-p}^*+\hc )\,.
\]
Using \eqref{eq:actionT} we get
\[ 
T^*_{\nu} \cK T_\nu - \left[ \cK + \sum_{p \in \L_+^*} p^2 \s_p^2 \right] = 2 \sum_{p\in \L^*_+} p^2 \s_p^2 a^*_p a_p +  \sum_{p \in \L^*_+} p^2  \big[\g_p\s_p ( a^*_p a^*_{-p}+\hc)  \big] := E_1 + E_2\,.
\]
From (\ref{eq:Adef}), $\langle \xi_\nu, \text{E}_2 \xi_\nu \rangle = 0$ ($\xi_\nu$ is a superposition of states with $3m$ particles, for $m \in \bN$). To bound the expectation of $E_1$ on $\xi_\nu$ we notice that $\langle\xi_\nu,a^*_pa_p\xi_\nu\rangle=0$ if $p\in \L^*_+\backslash(P_S\cup P_H)$. Moreover, proceeding as in \eqref{eq:sL-H1}, we have
\[ 
\sup_{p \in P_S} ( p^2 \s_p^2)     \leq  \sup_{p \in P_S: |p|\leq N^{\kappa/2}} N^{\kappa/2} |p|  
+ \sup_{p \in P_S : |p| > N^{\kappa/2}} \frac{N^{2\kappa}}{p^2}  \leq   C N^{\k}  
\]
while  
\[
 \sup_{p \in P_H} (p^2 \s_p^2) \leq \sup_{|p| \geq N^{1-\kappa-\eps}}   \frac{N^{2\k}}{p^2} \leq C N^{-2+4\kappa+2\eps}  \leq C N^\kappa 
\]
because, by assumption, $3\kappa-2+2\eps < 0$. Hence
\be \label{eq:E3}
|\bmedia{\xi_\nu, E_{1}\xi_\nu}|\leq C N^\k  \, \| \cN^{1/2}\xi_\nu \|^2 \,. 
\ee

We consider now the contribution from $ \cL_N^{(2,V)}$. Using again $\langle \xi_\nu , a_p^* a_{-p}^* \xi_\nu \rangle = 0$ for all $p \in \L_+^*$ and $\langle \xi_\nu , a_p^* a_{p} \xi_\nu \rangle = 0$ for all $p \in \L^*_+\backslash(P_S \cup P_H)$, a straightforward computation shows that 
\[ \begin{split} 
&\frac{\langle \xi_\nu, T^*_\nu \cL_N^{(2,V)} T_\nu \xi_\nu \rangle}{\| \xi_\nu \|^2} \\ &\hspace{2cm} =  \frac{N_0}{N} \sum_{p \in \L^*_+} N^\k \widehat V(p/N^{1-\k}) \g_p \s_p +  \frac{N_0}{N} \sum_{p \in \L^*_+} N^\k \big(\widehat V(0) +\widehat V(p/N^{1-\k}) \big) \s^2_p  \\ & \hspace{2.3cm} +  \frac{N_0}{N} \sum_{p\in P_S \cup P_H}  N^\k  \big[\widehat V(p/N^{\k-1})(\g_p+\s_p)^2 +\widehat V(0)  (\g^2_p + \s^2_p) \big]  \frac{\langle \xi_\nu, a^*_p a_p \xi_\nu \rangle}{\| \xi_\nu \|^2} \,.
\end{split} \]
With the bounds $\|\g_S \|^2_\io$, $\|\s_S\|^2_\io$,$\|\s_H\|_\io^2$,$\|\g_H\|_\io^2 \leq C N^{\e}$ from Lemma \ref{lm:nu-norms}, with (\ref{eq:E3}) and with the estimate $\langle \xi_\nu, \cN \xi_\nu \rangle \leq C N^{9\kappa/2 -2+\eps} \| \xi_\nu \|^2$ from Prop. \ref{prop:Abounds}, we conclude that 
\be \label{eq:E4}
\begin{split} 
\frac{\langle \xi_\nu, \cG^{(2)}_N \xi_\nu \rangle}{\| \xi_\nu \|^2}  \leq \; &\frac{\langle \xi_\nu, \cK \xi_\nu \rangle}{\| \xi_\nu \|^2}  + \sum_{p \in \L^*_+} p^2 \s_p^2 + \frac{N_0}{N} \sum_{p \in \L^*_+} N^\k \widehat V(p/N^{1-\k}) \g_p \s_p \\ &+  \frac{N_0}{N} \sum_{p \in \L^*_+} N^\k \big(\widehat V(0) +\widehat V(p/N^{1-\k}) \big) \s^2_p  + C N^{5\kappa/2-\eps}
\end{split} \ee
using again the condition $3\k - 2 +4\e < 0$. \newpage
 
Next, we study the contribution of $\cG_N^{(3)} = T_\nu^* \cL^{(3)}_N T_\nu$, with $\cL^{(3)}_N$ as in (\ref{eq:cL}). Recall the operator $\cC_N$, defined in (\ref{eq:cCG}). Taking into account the fact that $\xi_\nu$ is a superposition of vectors with $2m$ particles with momenta in $P_H$ and $m$ particles with momenta in $P_S$, for $m \in \bN$, we obtain that 
\[ \langle \xi_\nu, \cG_N^{(3)} \xi_\nu \rangle = \langle \xi_\nu , \cC_N \xi_\nu \rangle + \sum_{j=1}^3 \left[ \langle \xi_\nu, \text{F}_j \xi_\nu \rangle + \text{h.c.} \right] \]
with 
\[ \begin{split} 
\text{F}_1 &=  \frac{\sqrt{N_0}}{N} \sum_{p,r \in P_H : p+r \in P_S} N^\k \widehat V(r/N^{1-\k}) \g_{p+r} \s_p \s_r  \, a^*_{p+r} a^*_{-p} a^*_{-r}  \\
\text{F}_2 &=  \frac{\sqrt{N_0}}{N} \sum_{p \in P_H  , r \in P_S : p+r \in P_H} N^\k \left[ \widehat V(r/N^{1-\k}) + \widehat{V} (p / N^{1-\k}) \right] \g_{p+r} \s_p \s_r  \, a^*_{p+r} a^*_{-p} a^*_{-r}  \\
\text{F}_3 &= \frac{\sqrt{N_0}}{N}  \sum_{p \in P_H, r \in P_S : p+r \in P_H}  N^\k \left[ \widehat V(r/N^{1-\k}) + \widehat V(p /N^{1-\k}) \right] \sigma_{p+r} \g_p \g_r  \, a_{-p-r} a_{p} a_{r}  \,.
\end{split} \]
Using $\|a^*_{-r} (\cN+1)^{1/2}\xi_\nu\|\leq \|a_{-r} (\cN+1)^{1/2}\xi_\nu\| + \|(\cN+1)^{1/2}\xi_\nu\|$, 
we can bound
\[ \begin{split} 
\langle \xi_\nu , \text{F}_1 \xi_\nu \rangle \leq \; &C N^{\kappa -1/2} \| \gamma_S \|_\infty \sum_{p,r \in P_H} |\s_r| |\s_p | \, \| a_{p+r} a_{-p} (\cN + 1)^{-1/2}  \xi_\nu \| \\ &\hspace{4cm} \times \left[ \| a_{-r} (\cN + 1)^{1/2} \xi_\nu \| + \| (\cN + 1)^{1/2} \xi_{\nu} \| \right] \\ \leq \; &C N^{\kappa -1/2} \| \gamma_S \|_\infty \| \s_H \|_\infty \| \s_H \| \, \| (\cN + 1)^{1/2} \xi_\nu \| \| (\cN+ 1) \xi_\nu \| \\ &+ C N^{\kappa -1/2} \| \gamma_S \|_\infty \| \s_H \|^2 \, \| (\cN + 1)^{1/2} \xi_\nu \|^2\,.  \end{split}  \]
With Lemma \ref{lm:nu-norms} and Prop. \ref{prop:Abounds}, we obtain 
\[ \frac{\langle  \xi_\nu , \text{F}_1 \xi_\nu \rangle}{\| \xi_\nu \|^2}  \leq CN^{37\k /4 - 4 +5\e/2} + C N^{17\kappa/2 - 7/2 + 5\eps/2}  \leq C N^{5\k/2} \cdot N^{21\k /4 - 7/2 + 3\eps/2}  \]
from the assumption that $3\k -2 +4\eps < 0$. Similarly, we find 
\[  \begin{split}  \langle \xi_\nu , \text{F}_2 \xi_\nu \rangle \leq \; & C N^{\kappa -1/2} \| \gamma_H \|_\infty \| \s_S \|_\infty \| \s_H \| \, \| (\cN+  1)^{1/2} \xi_\nu \| \| (\cN+ 1) \xi_\nu \| \\ &+ C N^{\kappa -1/2} \| \gamma_H \|_\infty \| \s_S \| \| \s_H \| \, \| (\cN + 1)^{1/2} \xi_\nu \|^2 \\  \leq \; &C \big[ N^{37\k/4-4+3\e} + N^{31\k/4-3+3\e/2} \big] \| \xi_\nu \|^2 \leq C N^{5\k/2} \cdot N^{21\kappa/4 - 3 + 3 \eps /2}  \| \xi_\nu \|^2 \end{split}  \]
and also 
\[ \begin{split}  
\langle \xi_\nu , \text{F}_3 \xi_\nu \rangle \leq \; & C N^{\kappa -1/2} 
\| \gamma_H \|_\infty  \| \g_S \|_\infty \| \s_H \| \, \| (\cN+  1)^{1/2} \xi_\nu \| \| (\cN+ 1) \xi_\nu \| \\ &+ C N^{\kappa -1/2} \| \gamma_H \|_\infty \| \gamma_S \| \| \s_H \| \, \| (\cN + 1)^{1/2} \xi_\nu \|^2 \\  \leq \; &C N^{5\kappa/2} \cdot N^{21\kappa/4 - 3 + 3 \eps}  \| \xi_\nu \|^2 \, . \end{split}  \]
Summarizing, we have 
\begin{equation}\label{eq:fin-G3} \frac{\langle \xi_\nu, \cG_N^{(3)} \xi_\nu \rangle}{ \| \xi_\nu \|^2} \leq \frac{\langle \xi_\nu , \cC_N \xi_\nu \rangle}{\| \xi_\nu \|^2}  + C N^{5\kappa/2} \cdot N^{21\kappa/4 - 3 + 3 \eps} \, . \end{equation} 

Finally, let us consider $\cG_N^{(4)} = T_\nu^* \cL_N^{(4)} T_\nu$. 
We decompose $\langle \xi_\nu , \cG_N^{(4)} \xi_\nu \rangle = \sum_{j=1}^3 \langle \xi_\nu , \text{G}_j \xi_\nu \rangle $
with
\[ \begin{split}
\text{G}_1  &=  \frac 1 {2N} \sum_{\substack{r \in \L^*,\, p,q \in \L^*_+ \\ -r \neq q,p}} N^\k \widehat V(r/N^{1-\k}) \g_p \g_q \g_{p+r} \g_{q+r}  a^*_{p+r} a^*_q a_{p} a_{q+r}  \\
\text{G}_2 & =  \frac 1 {2N} \sum_{\substack{r \in \L^*,\, p,q \in \L^*_+ \\ r \neq q,-p}} N^\k \widehat V(r/N^{1-\k}) \big( \g_{p+r} \s_q a^*_{p+r} a_{-q} + \s_{p+r} \g_q a_{-p-r} a^*_q \big) \\[-0.5cm]
&  \hskip5.5cm\times \big( \g_{p} \s_{q+r} a_{p}a^*_{-q-r} + \s_{p} \g_{q+r} a^*_{-p} a_{q+r}\big)\\[0.2cm]
\text{G}_3 & = \frac 1 {2N} \sum_{\substack{r \in \L^*,\, p,q \in \L^*_+ \\ r \neq q,-p}} N^\k \widehat V(r/N^{1-\k})  \s_p \s_q \s_{p+r} \s_{q+r} a_{p+r} a_{q}  a^*_{p} a^*_{q+r}\,.
\end{split}\]
To estimate contributions from $\text{G}_3$, we arrange terms in normal order. We find 
\[ \begin{split} 
\text{G}_3 = \; & \frac 1 {2N} \sum_{\substack{r \in \L^*,\, p,q \in \L^*_+ \\ -r \neq q,p}} N^\k \widehat V(r/N^{1-\k})  \s_p \s_q \s_{p+r} \s_{q+r} a^*_{p} a^*_{q+r} a_{p+r} a_{q} \\ &+  \frac 1 {2N} \sum_{\substack{r \in \L^*,\, p \in \L^*_+ \\ p \neq -r}} N^\k \widehat V(r/N^{1-\k})  \s_p^2 \s^2_{p+r} \big( a^*_p a_p + a^*_{p+r}a_{p+r}\big) \\ &+  \frac 1 {N} \sum_{\substack{p,q \in \L^*_+ }} N^\k \widehat V(0)  \s_p^2 \s^2_{q}  a^*_p a_p  \\ &+  \frac 1 {2N} \sum_{\substack{r \in \L^*,\, p \in \L^*_+ :\\ p\not = -r}} N^\k \widehat V(r/N^{1-\k}) \s^2_p \s^2_{p+r} +  \frac 1 {2N} \sum_{\substack{p,q \in \L^*_+}} N^\k \widehat V(0) \s^2_p \s^2_{q} \,.
\end{split}\]
Since $a_p\,\xi_\nu=0$ if $p\in \L^*_+\backslash(P_S\cup P_H)$ and $\|\s_H\|_\io\leq \|\s_S\|_\io$ we find, by Cauchy-Schwarz,  
\begin{equation}\label{eq:G3-fin}  \begin{split}  \langle \xi_\nu, \text{G}_3 \xi_\nu \rangle &\leq C N^{\kappa -1} \left[ \| \s_S \|_\infty^2 \| \sigma \|^2 \| (\cN+1) \xi_\nu \|^2 + \| \s \|^4 \| \xi_\nu \|^2 \right] \\ &\leq C N^{5\kappa/2} \cdot \max \{ N^{-\eps} , N^{9\kappa -5 + 3\eps} \}  \| \xi_\nu \|^2 \end{split} \end{equation} 
using Prop. \ref{prop:Abounds} and $3\k -2 + 4\eps < 0$. 
We proceed similarly for $\text{G}_2$. Through normal ordering, we get 
\[ \begin{split} 
\text{G}_2 = \; &\frac 1 {N} \sum_{\substack{r \in \L^*,\, p,q \in \L^*_+ \\ -r \neq q,p}} N^\k \widehat V(r/N^{1-\k})  \g_p \g_{p+r}\s_q \s_{q+r} a^*_{p+r} a^*_{-q-r} a_{p} a_{-q} \\ &+  \frac 1 {N} \sum_{\substack{r \in \L^*,\, p,q \in \L^*_+ \\ -r \neq q,p}} N^\k \widehat V(r/N^{1-\k})  \g_{p+r}  \g_{q+r}\s_p \s_{q} a^*_{p+r} a^*_{-p} a_{-q} a_{q+r} \\ &+  \frac 1 {N} \sum_{\substack{p,q \in \L^*_+ }} N^\k \widehat V(0)  \g_p^2 \s^2_q  a^*_p a_p  + \frac 1 {N} \sum_{\substack{r \in \L^*,\,p \in \L^*_+ }} N^\k \widehat V(r/N^{1-\k})  \g_p^2 \s^2_{p+r}  a^*_p a_p \\ &+ \frac 2 {N} \sum_{\substack{r \in \L^*,\,p \in \L^*_+ }} N^\k \widehat V(r/N^{1-\k})  \g_p \s_p \g_{p+r} \s_{p+r}  a^*_p a_p  \\ &+ \frac 1 {2N} \sum_{\substack{p,r \in \L^*_+}} N^\k \widehat V(r/N^{1-\k})  \g_p \g_{p+r} \s_p \s_{p+r}   \,. \end{split} \]
Keeping the last contribution intact and estimating the term on the fourth line distinguishing the two cases $(p+r) \in P_S$ and $(p+r) \in P_H$, we arrive at 
\[ \begin{split}  \langle \xi_\nu , \text{G}_2 \xi_\nu \rangle \leq \; &\frac 1 {2N} \sum_{\substack{p,r \in \L^*_+}} N^\k \widehat V(r/N^{1-\k})  \g_p \g_{p+r} \s_p \s_{p+r}  \| \xi_\nu \|^2  \\ &+ C N^{\kappa -1} \| \g_S \|_\infty^2 \| \s \|^2 \| (\cN+1) \xi_\nu \|^2  \\ &+ C N^{\kappa -1} \| \g_{S \cup H} \|_\infty \| \s_{S \cup H}  \|_\infty  \\ & \hspace{2cm} \times \left[ \| \gamma_S \s_S \|_1 + \| \g_H \|_\infty \sup_p \sum_{r \in \L^*} \widehat{V} (r/ N^{1-\kappa}) |\eta_{p+r}| \right] \| \cN^{1/2} \xi_\nu \|^2    \,. \end{split} \]
With the bounds in Lemma \ref{lm:nu-norms} 
and in Prop. \ref{prop:Abounds} and with (\ref{eq:intVeta}), we conclude that 
\begin{equation}\label{eq:G2-fin}  \langle \xi_\nu , \text{G}_2 \xi_\nu \rangle \leq  \frac 1 {2N} \sum_{\substack{p,r \in \L^*_+}} N^\k \widehat V(r/N^{1-\k})  \g_p \g_{p+r} \s_p \s_{p+r}  \| \xi_\nu \|^2 + C N^{5\kappa /2} \cdot N^{9\kappa -5 + 3\eps} \| \xi_\nu \|^2\,. \end{equation}
Finally, we consider $\text{G}_1$. Recalling that $a_p \xi_\nu = 0$ if $p \in \Lambda^*_+ \backslash (P_S \cup P_H)$ and observing that $\langle \xi_\nu, a^*_{p+r}a^*_qa_pa_{q+r}\xi_\nu\rangle\neq 0$ only if the operator $a^*_{p+r}a^*_q a_p a_{q+r}$ preserves the number of particles in $P_S$ and in $P_H$, we arrive at 
\[ \begin{split} \langle \xi_\nu, \text{G}_1 \xi_\nu \rangle \leq \; &\frac 1 {2N} \sum_{\substack{r \in \L^*,\, p,q \in P_H: \\p+r,q+r\in P_H}} N^\k \widehat V(r/N^{1-\k})  \g_p \g_q \g_{p+r} \g_{q-r}  \langle \xi_\nu , a^*_{p+r} a^*_q a_{p} a_{q+r} \xi_\nu \rangle \\ &+ C N^{\kappa -1 } \| \g_{S \cup H} \|_\infty^2 \| \gamma_S \|^2 \| (\cN + 1) \xi \|^2\,.  \end{split} \]
With $|\g_p \g_q \g_{p+r} \g_{q+r} - 1| \leq C \| \eta_H \|_\infty^2$ for all $p,q \in P_H$, with $(p+r), (q+r) \in P_H$, and using the estimate (see the proof of (\ref{eq:intVeta}))  
\[ \sup_{p\in \L^*}  \sum_{r \in \L^*_+ : r \not = p}  \frac{N^\k |\widehat{V} (r/N^{1-\k})|}{|p-r|^2} \leq C N \]
we conclude that 
\[ \langle \xi_\nu, \text{G}_1 \xi_\nu \rangle \leq \langle \xi_\nu, \cV_N^{(H)} \xi_\nu \rangle + C \| \eta_H \|_\infty^2 \| \cN^{1/2} \cK^{1/2} \xi_\nu \|^2 + C N^{\kappa -1 } \| \g_{S \cup H} \|_\infty^2 \| \gamma_S \|^2 \| \cN \xi \|^2  \]
with $\cV_N^{(H)}$ defined as in (\ref{eq:cHN}). With Lemma \ref{lm:nu-norms} and Prop.  \ref{prop:Abounds}, we find (using the assumption $3\k - 2 + 4 \eps < 0$) 
\[ \langle \xi_\nu, \text{G}_1 \xi_\nu \rangle \leq  \langle \xi_\nu, \cV_N^{(H)} \xi_\nu \rangle + C N^{5\kappa/2} \cdot N^{9 \kappa - 5 + 6\eps} \| \xi_\nu \|^2\,. \]
With (\ref{eq:G3-fin}) and (\ref{eq:G2-fin}), we have shown that 
\[ \frac{\langle \xi_\nu, \cG_N^{(4)} \xi_\nu \rangle}{\| \xi_\nu \|^2} \leq \frac{\langle \xi_\nu, \cV_N^{(H)}  \xi_\nu \rangle}{\| \xi_\nu \|^2} + C N^{5\kappa/2} \cdot \max \{ N^{-\eps}, N^{9\kappa - 5 + 6 \eps} \}\,.  \]

Combining the last bound with (\ref{eq:E4}) and (\ref{eq:fin-G3}), we obtain  
\[ \begin{split}  \frac{\langle \xi_\nu, \cG_N \xi_\nu \rangle}{\| \xi_\nu \|^2} \leq \; &\wt{C}_N + \frac{\langle \xi_\nu, (\cK  +\cV_N^{(H)} + \cC_N) \xi_\nu \rangle}{\| \xi_\nu \|^2}  \\ &+ C N^{5\kappa/2} \cdot  \max \{ N^{-\eps}, N^{9\k - 5 + 6\e} , N^{21 \kappa /4 - 3 + 3 \eps } \}  \end{split} \]
where we defined 
\be \begin{split} \label{eq:constG-1}
\wt{C}_{N} =\; &  \frac{N_0^2 }{2N} N^\k \widehat{V}(0) +\sum_{p\in {\L}^*_+}p^2\s_p^2 + \frac{N_0}{N}\sum_{p\in \L^*_+} N^\k \big( \widehat{V}(p/N^{1-\k}) + \widehat{V}(0) \big)\s_p^2\\
&  + \frac{N_0}{N}\sum_{p\in \L^*_+}N^\k \widehat{V}(p/N^{1-\k})\s_p\g_p +\frac{1}{2N}\sum_{\substack{p,r\in \L^*_+\\r\neq p}}N^\k \widehat{V}(r/N^{1-\k})\s_p\s_{p+r}\g_p\g_{p+r}\,.
\end{split}\ee
Inserting $N_0 = N - \| \sigma_L \|^2$ and recalling from Lemma \ref{lm:nu-norms} that $\| \sigma_L \|^2 \leq C N^{3\k /2}$ and $\| \s_{L^c} \|^2  \leq C N^{3\k /2 - \eps}$, we obtain 
$\wt{C}_N = C_{\cG_N} + \cO (N^{5\k /2 - \eps})$, with $C_{\cG_N}$ as defined in (\ref{eq:constantG}) (with the assumption $3\k - 2 + 4\eps < 0$). To handle the first term on 
the second line of (\ref{eq:constG-1}), we used that $|\s_p \g_p - \eta_p| \leq C \eta_p^3 
\leq C N^{3\k}/|p|^6$, for $p \in P_L^c$. This completes the proof of Prop. \ref{prop:calG}.


\section{Cubic conjugation}  
  \label{sec:cubicconj}  

In this section we prove Prop. \ref{prop:Abounds} and Prop. \ref{prop:cubic-contr}, which is a conequence of the following lemma. 
\begin{lemma} \label{lm:cubic}
Let $A_\nu$ be defined in \eqref{eq:Adef}, and $\cK$, $\cV_N^{(H)}$ and $\cC_N$ be defined in \eqref{eq:cHN}
and \eqref{eq:cCG} respectively. Then, for $\xi_\nu=e^{A_\nu}\O$, 
\begin{align}
 \frac{\langle \xi_\nu, \cK \xi_\nu \rangle}{\| \xi_\nu \|^2} \leq\; & \frac 2 {N} \sum_{\substack{v \in P_S, r \in P_H:\\r+v \in P_H}}  r^2\eta_r (\eta_r+\eta_{r+v})\s_v^2 +  \cE ,
  \label{eq:K}  \\
\frac{\langle \xi_\nu, \cC_N \xi_\nu \rangle}{\| \xi_\nu \|^2} \leq \; & \frac 2 {N} \sum_{\substack{v \in P_S, r \in P_H :\\ r+v \in P_H}} N^\k \widehat V(r/N^{1-\k})(\eta_r +\eta_{r+v}) \s_v^2 + \cE , \label{eq:C} \\
\frac{\langle \xi_\nu,  \cV_N^{(H)} \xi_\nu \rangle}{\| \xi_\nu \|^2} \leq \;& \frac 1 {N^2} \sum_{\substack{v \in P_S, r \in P_H: \\ r+v \in P_H}} \, \big( N^\k \widehat V(\cdot/N^{1-\k}) \ast \eta \big)_r  (\eta_r+\eta_{r+v})  \s_v^2 + \cE , \label{eq:V}
\end{align}
with 
\[
\cE \leq C N^{5\k/2}  \cdot \max \{ N^{-\e},\, N^{12\k-7 +5\e}  \} 
\]
for all $\k \in (1/2 ;2/3)$, $\e>0$ so small that $3\k - 2 + 4 \e < 0$ and $N$ large enough. 
\end{lemma}

With Lemma \ref{lm:cubic}, we can immediately show Prop. \ref{prop:cubic-contr}. 

\begin{proof}[Proof of Proposition \ref{prop:cubic-contr}] From Lemma \ref{lm:cubic} we have 
\[  \begin{split} \label{eq:D1}
  &\frac{\langle \xi_\nu, (\cK + \cV_N^{(H)} + \cC_{N} ) \xi_\nu \rangle}{\| \xi_\nu \|^2} \\
  & \leq \frac 2 N \sum_{v \in P_S} \s_v^2  \sum_{\substack{r \in P_H: \\ r+v \in P_H}}\Big[ r^2 \eta_r + N^\k   \widehat V(r/N^{1-\k})  +  \frac {N^\k} {2N}  \big(\widehat V(\cdot/N^{1-\k}) \ast \eta  \big)_r  \Big] (\eta_r+\eta_{r+v})
 + \cE \end{split} 
 \]
 with $\cE \leq CN^{5\k/2} \cdot \max \{ N^{-\eps}, N^{12\k -7 + 5\eps} \}$. With 
 the scattering equation \eqref{eq:scatt}, we obtain 
\[
\frac{\langle \xi_\nu, (\cK + \cV_N^{(H)} + \cC_{N} ) \xi_\nu \rangle}{\| \xi_\nu \|^2} \leq  \frac1N\sum_{v\in P_S}\s_v^2\sum_{\substack{r \in P_H: \\r+v\in P_H}} N^\k  \widehat V(r/N^{1-\k})(\eta_r+\eta_{r+v})+ \cE' 
\]
with 
\[  \cE' \leq \; \frac 1 N \sum_{v \in P_S} \s_v^2 \sum_{\substack{r \in P_H:\, r+v \in P_H}} N^{3-2\k} \l_{\ell} (\widehat \chi_\ell \ast \widehat f_{N})_r \eta_r  + \cE \,. 
\]
Using $|N^{3-3\k}\l_{\ell}|\leq C$ and $\|\widehat{\chi}_\ell \ast \widehat f_{N}\|\leq C$, we conclude 
\[ \begin{split} 
&\frac{\langle \xi_\nu, (\cK + \cV_N^{(H)} + \cC_{N} ) \xi_\nu \rangle}{\| \xi_\nu \|^2} \\ &\hspace{.3cm} \leq  
\frac1N\sum_{v\in P_S}\s_v^2\sum_{\substack{r \in P_H: \\r+v\in P_H}} N^\k  \widehat V(r/N^{1-\k})(\eta_r+\eta_{r+v}) + C N^{5\k/2} \cdot \max \{ N^{-\eps}, N^{12\k -7 + 5\eps} \}\,. \end{split}  \]
Finally, with (\ref{eq:intVeta}) and the expression \eqref{eq:sL2} for $\s^2_v$, we can extend the sum over $v \in P_S$ to a sum over all $v \in P_L$, without changing the size of the error. This completes the proof of Prop. \ref{prop:cubic-contr}.
\end{proof} 

We still have to show Prop. \ref{prop:Abounds} and Lemma \ref{lm:cubic}.

\subsection{Expectation of the particle number and kinetic energy} 

In this section we prove \eqref{eq:K} and Prop \ref{prop:Abounds}. We start by computing the expectation $\langle \xi_\nu , \cK \xi_\nu \rangle$. We proceed as we did in (\ref{eq:normA0})-(\ref{eq:normA}) to compute $\| \xi_\nu \|^2$. With $\cK a_{r+v}^* a_{-r}^* a_{-v}^* = a_{r+v}^* a_{-r}^* a_{-v}^* (\cK + (r+v)^2 + r^2 + v^2)$ we obtain 
\[ \begin{split} \label{eq:AcKA-1}
\bmedia{\xi_{\nu},\cK \xi_{\nu}} = &\; \sum_{m\geq 1}\frac1{2^m (m-1)!}\frac1{N^m}  \sum_{\substack{ v_1 \in P_S, r_1 \in P_H : \\  r_1 +v_1 \in P_H  } } \cdots \sum_{\substack{ v_m \in P_S, r_m \in P_H : \\  r_m +v_m \in P_H  } }   \theta \big( \{ r_j, v_j \}_{j=1}^{m} \big) \\&\hspace{3cm} \times [r_m^2+v_m^2+(r_m+v_m)^2]\, \prod_{i=1}^m (\eta_{r_i} + \eta_{r_i + v_i})^2 \sigma_{v_i}^2  \, .
\end{split}\]
with the cutoff $\theta$ introduced in (\ref{eq:theta}). Since all terms are positive, we can find an upper bound for $\langle \xi_\nu, \cK \xi_\nu \rangle$ by replacing $\theta(\{r_j,v_j\}_{j=1}^m)$ with $\theta(\{r_j,v_j\}_{j=1}^{m-1})$, removing conditions 
involving momenta with index $m$. Recalling \eqref{eq:normA}, we find 
\[ \begin{split} \bmedia{\xi_{\nu},\cK \xi_{\nu}} &\leq \frac{1}{2N} \sum_{\substack{v\in P_S,r\in P_H:\\r+v\in P_H}}[r^2+v^2+(r+v)^2]  (\eta_r + \eta_{r+v})^2 \s_v^2 \,\|\xi_\nu\|^2 \\
 &\leq \frac{2}{N} \sum_{v \in P_S , r\in P_H} r^2 \eta_r(\eta_r + \eta_{r+v}) \s_v^2 \,\|\xi_\nu\|^2 + \cE \end{split} \]
with (using Lemma \ref{lm:nu-norms} and the assumption $3\k -2 +4\e < 0$) 
\[\begin{split} \frac{\cE}{\|\xi_\nu\|^2} =& \frac2N \sum_{\substack{v\in P_S,r\in P_H : \\ r+v \in P_H}} \hspace{-.3cm} (v ^2+r\cdot v) \eta_r  (\eta_r+\eta_{r+v})\s_v^2 \\&\leq \frac C N( \| \s_S \|_{H^1}^2 \| \eta_H \|^2  +\|\s_S\|\|\eta_H\|\|\s_S\|_{H^1}\|\eta_H\|_{H^1})\leq C N^{4\k -1+\e}\leq C N^{5\k/2-\eps}\,. \end{split}\]
This proves \eqref{eq:K}. In particular, \eqref{eq:K} implies, together with Lemma \ref{lm:nu-norms}, that 
\begin{equation}\label{eq:K-fin}  \frac{\langle \xi_\nu , \cK \xi_\nu \rangle}{\| \xi_\nu \|^2}  \leq C N^{-1} \| \eta_H \|_{H^1}^2 \| \s_S \|^2 \leq C N^{5\k /2} \end{equation} 
which shows (\ref{eq:cKcN}) with $j=1$ in Prop. \ref{prop:Abounds}. 

Analogously, we find 
\[\begin{split}  \langle \xi_\nu, \cK \cN \xi_\nu \rangle & \leq  \sum_{m\geq 1}\frac{3m}{2^m (m-1)!}\frac1{N^m}  \sum_{\substack{ v_1 \in P_S, r_1 \in P_H : \\  r_1 +v_1 \in P_H  } } \cdots \sum_{\substack{ v_m \in P_S, r_m \in P_H : \\  r_m +v_m \in P_H  } }   \theta \big( \{ r_j, v_j \}_{j=1}^{m} \big) \\&\hspace{3cm} \times [r_m^2+v_m^2+(r_m+v_m)^2]\, \prod_{i=1}^m (\eta_{r_i} + \eta_{r_i + v_i})^2 \sigma_{v_i}^2  \, .
\end{split}\]
Writing $m = 1+ (m-1)$, and bounding $\theta( \{ r_j, v_j \}_{j=1}^{m} )$ by $\theta( \{ r_j, v_j \}_{j=1}^{m-2} )$, we obtain 
\[\begin{split}  \langle \xi_\nu, \cK \cN \xi_\nu \rangle \leq  \; & 3 \langle \xi_\nu , \cK \xi_\nu \rangle + 
\frac{3}{4N^2} \sum_{r,r' \in P_H, v,v' \in P_S} [ r^2 + v^2 + (r+v)^2]  \\ &\hspace{5cm} \times (\eta_r + \eta_{r+v})^2 (\eta_{r'} + \eta_{r' + v'})^2 \s_v^2 \s_{v'}^2  \| \xi_\nu \|^2\,. \end{split} \]
With (\ref{eq:K-fin}) and with the bounds for $\| \eta_H \|_{H^1}^2, \| \eta_H \|^2, \| \sigma_S \|^2$ from Lemma \ref{lm:nu-norms}), we find 
\begin{equation}\label{eq:KN2-fin} \frac{\langle \xi_\nu, \cK \cN \xi_\nu \rangle}{\| \xi_\nu \|^2} \leq C N^{5\k /2} \cdot N^{9\k/2 - 2 + \eps} \end{equation} 
which shows (\ref{eq:cKcN}) with $j=2$. 

To show \eqref{eq:xiAxi} we observe that, by \eqref{eq:Adef}, the operator $A_\nu$ only creates particles with momenta in $P_S \cup P_H$ and for each particle with momentum in $P_S$, it creates two particles with momenta in $P_H$. Since $|p| > N^{1-\k-\eps}$ for all $p \in P_H$, we find, by (\ref{eq:K-fin}), 
\[\begin{split} 
\langle \xi_\nu, \cN \xi_\nu \rangle  = &\sum_{p \in P_S \cup P_H} \langle \xi_\nu, a^*_p a_p \xi_\nu \rangle = \frac{3}{2} \sum_{p \in P_H}  \langle \xi_\nu , a^*_p a_p \xi_\nu \rangle \\ &\leq C N^{-2+2\k+2\e} \langle \xi_\nu, \cK \xi_\nu \rangle \leq N^{9\k/2 - 2 + 2\eps} \| \xi_\nu \|^2 \end{split} \]
proving (\ref{eq:xiAxi}) for $j=1$. Analogously, we find 
\[\begin{split} 
\langle \xi_\nu, \cN^2 \xi_\nu \rangle = &\; \sum_{p\in P_S\cup P_H} \langle \cN^{1/2}\xi_\nu, a^*_p a_p \cN^{1/2}\xi_\nu \rangle  \\ \leq \; &\frac32 \sum_{p\in P_H} \langle \cN^{1/2}\xi_\nu, a^*_p a_p \cN^{1/2} \xi_\nu \rangle \leq C N^{-2+2\k+2\eps} \langle \xi_\nu, \cK \cN \xi_\nu \rangle \,.
 \end{split} \]
By (\ref{eq:KN2-fin}), we obtain (\ref{eq:xiAxi}) with $j=2$. This completes the proof of Prop. \ref{prop:Abounds}.

\subsection{Expectation of the cubic term}

The goal of this section is to show \eqref{eq:C}.  From \eqref{eq:cCG}, we have (using the reality of $\eta_p, \g_p, \s_p$) 
\[\begin{split}
	&\langle\xi_\nu,\cC_N\xi_\nu\rangle\\
	&=2\frac{\sqrt{N_0}}{ N}\sum_{m\geq 1}\frac1{m!(m-1)!}\sum_{\substack{p,r \in P_H\\ p+r \in P_S}} N^\k \widehat V(r/N^{1-\k})\, \sigma_{p+r} \g_p\g_r \langle A_\nu^{m} \xi_\nu,  \,a^*_{p+r} a^*_{-p} a^*_{-r}  A_\nu^{m-1}\xi_\nu\rangle.
\end{split}\]
Proceeding as in the previous section, we get 
\[\label{eq:expC-2}\begin{aligned} 
& \bmedia{\xi_{\nu},\cC_N\xi_{\nu}}\\
&= 2 \sqrt{\frac{N_0}{N}}\,\sum_{m\geq 1}\frac1{2^{m-1} (m-1)!}\frac1{N^m}\sum_{\substack{ v_1 \in P_S, r_1 \in P_H : \\  r_1 +v_1 \in P_H  } } \cdots  \sum_{\substack{ v_m \in P_S, r_m \in P_H : \\  r_m +v_m \in P_H  } } \hskip -0.5cm  \theta\big( \{r_j, v_j \}_{j=1}^{m} \big)  \\ 
& \hskip 0.5cm \times N^\k  \widehat V(r_m/N^{1-\k})\,\big(\eta_{r_m} +\eta_{r_m+v_m}\big)\g_{r_m}\g_{r_m+v_m} \s^2_{v_m}\,  \prod_{i=1}^{m-1} (\eta_{r_i} + \eta_{r_i + v_i})^2 \s_{v_i}^2  .
\end{aligned}\]
To reconstruct the norm $\| \xi_\nu \|^2$ on the r.h.s. we need to free the momenta with index $m$. To this end, we recall the defintion (\ref{eq:theta}) to write 
\be \label{eq:remove}
\theta\big( \{r_j, v_j \}_{j=1}^{m} \big) = \theta\big( \{r_j, v_j \}_{j=1}^{m-1} \big)\, \theta_m\big( \{ r_j, v_j \}_{j=1}^{m} \big)
\ee
with 
\[ \label{eq:thetam}
 \theta_m\big( \{r_j, v_j \}_{j=1}^{m} \big)=  \prod_{i,j=1}^{m-1}\prod_{\substack{p_i,p_j,p_m:\\p_\ell\in\{-r_\ell,r_\ell+v_\ell\}}}\d_{p_i\neq -p_j+v_m}\d_{-p_m+v_i\neq p_j}
\]
collecting all conditions involving $\{r_m, v_m\}$. Writing $\theta_m = 1 + [ \theta_m - 1]$, we split $\langle \xi_\nu, \cC_N \xi_\nu \rangle = I_\cC + J_{\cC}$ with (recall the expression (\ref{eq:normA}) for $\| \xi_\nu \|^2$)  
\begin{equation*} \label{eq:I-cC}
I_\cC = \; 2\,\sqrt{\frac{N_0}{N}}\,\sum_{\substack{ v \in P_S,\, r \in P_H: \\r+v \in P_H}}  N^{\k-1}\widehat V(r/N^{1-\k})\,\big(\eta_{r}+\eta_{r+v}\big)\g_{r}\g_{r+v} \s^2_{v}\, \| \xi_\nu\|^2
\end{equation*} 
and
\[\label{eq:JcC}\begin{aligned}
J_\cC 	=\; & 2\sqrt{\frac{N_0}{N}}\sum_{m\geq 1}\frac1{2^{m-1} (m-1)!}\frac1{N^m}  \sum_{\substack{ v_1 \in P_S, r_1 \in P_H : \\  r_1 +v_1 \in P_H  } } \cdots  \sum_{\substack{ v_m \in P_S, r_m \in P_H : \\  r_m +v_m \in P_H  } } \theta\big( \{r_j, v_j \}_{j=1}^{m-1} \big) \\
&   \times \Big[ \theta_m\big( \{r_j, v_j \}_{j=1}^{m} \big) -1 \Big]    N^\k\widehat V(r_m/N^{1-\k})\,\big(\eta_{r_m} +\eta_{r_m+v_m}\big)\g_{r_m}\g_{r_m+v_m} \s^2_{v_m}  \\ &\hspace{8cm} \times  
 \prod_{i=1}^{m-1} (\eta_{r_i} + \eta_{r_i + v_i})^2 \s_{v_i}^2 \,.
\end{aligned}\]
With $|\sqrt{N_0/N} - 1| \leq C \| \s_L \|^2/N$ and $|\g_r \g_{r+v} - 1| \leq C N^{2\kappa} / |r|^4$ for all $r \in P_H, v\in P_S$, we obtain (using (\ref{eq:intVeta}) and the assumption $3\kappa - 2 + 4\eps < 0$) that 
\begin{equation}\label{eq:IC-fin}  \frac{I_\cC}{\| \xi_\nu \|^2} \leq \frac{2}{N} \sum_{\substack{ v \in P_S,\, r \in P_H: \\r+v \in P_H}}  N^{\k}\widehat V(r/N^{1-\k})\,\big(\eta_{r}+\eta_{r+v}\big) \s^2_{v}  + C N^{5\k /2 - \eps}\, . \end{equation} 

To complete the proof of (\ref{eq:C}), we focus now on the error term $J_\cC$. We observe that  
\be \begin{split} \label{eq:thetam-one}
|  \theta_m\big( \{r_j, v_j \}_{j=1}^{m} \big) -1 |  \leq& \sum_{j=1}^{m-1} \Big[ \d_{v_j , v_m} + \sum_{ \substack{ p_m \in \{ -r_m, r_m+v_m \} \\p_j \in \{ -r_j, r_j+v_j\} }}\d_{p_m, p_j} \Big]  \\
& + \sum_{\substack{j,k=1\\j\neq k}}^{m-1} 
\Big[ \sum_{ \substack{ p_j \in \{ -r_j, r_j+v_j \} \\ p_k \in \{ -r_k, r_k+v_k\} }} \d_{v_m, p_j+p_k}
+  \sum_{\substack{ p_m \in \{ -r_m, r_m+v_m \} \\ p_j \in \{ -r_j, r_j+v_j \} }}\d_{p_m, -p_j+v_k}   \Big] \,.
\end{split}\ee
We bound $|J_\cC| \leq \text{X}_1+ \text{X}_2$, with $\text{X}_1$ denoting the contribution arising from the first term on the r.h.s. of (\ref{eq:thetam-one}) (this term involves two indices, $m$ and $j$), and $\text{X}_2$ indicating the contribution from the second term on the r.h.s. of (\ref{eq:thetam-one}) (this term involves three indices, $m,j,k$). We can estimate
\[ \begin{split} 
\text{X}_1 \leq \; &C \sum_{m\geq 2}\frac1{2^{m-2} (m-2)!}\frac1{N^m}\sum_{\substack{ v_1 \in P_S, r_1 \in P_H : \\  r_1 +v_1 \in P_H  } } \cdots  \sum_{\substack{ v_m \in P_S, r_m \in P_H : \\  r_m +v_m \in P_H  } }  \theta\big( \{r_j, v_j \}_{j=1}^{m-1} \big) \\ &\times N^\k |\widehat V(r_m/N^{1-\k})| \,\big| \eta_{r_m} +\eta_{r_m+v_m}\big| |\g_{r_m}| |\g_{r_m+v_m}| \s^2_{v_m}  \prod_{i=1}^{m-1} (\eta_{r_i} + \eta_{r_i + v_i})^2 \s_{v_i}^2  \\ & \times \Big[ \d_{v_m, v_{m-1}} +\sum_{\substack{p_{m-1},p_m:\\p_\ell\in \{-r_\ell,r_\ell+v_\ell\}}}\hspace{-0.3cm}\d_{p_m, p_{m-1}}\Big]  \,.
\end{split} \]
With $\theta\big( \{r_j, v_j \}_{j=1}^{m-1} \big) \leq  \theta\big( \{r_j, v_j \}_{j=1}^{m-2} \big)$, we reconstruct $\| \xi_\nu \|^2$. Since $\| \g_H \|_\infty \leq C$, we end up with 
\[ \begin{split} 
\frac{\text{X}_1}{\| \xi_\nu \|^2}  \leq\; & \frac{C}{N^2} \sum_{r , r' \in P_H, v,v' \in P_S} N^\kappa |\widehat{V} (r/N^{1-\k}) |\eta_r + \eta_{r+v}| |\eta_{r'} + \eta_{r'+v'}|^2 \sigma_v^2 \sigma_{v'}^2 \\ &\hspace{6cm} \times \Big[ \delta_{v,v'} + \sum_{\substack{p \in \{-r, r+v\} \\ p' \in \{ -r', r' +v' \}}} \delta_{p,p'} \Big] \\
\leq \; &C N^{2\k-2} \| \s_S \|_\infty^2 \| \s_S \|^2 \| \eta_H \|^2 \sum_{r \in P_H} \frac{|\widehat{V} (r/N^{1-\k})|}{r^2} + C N^{4\k -2} \| \s_S \|^4 \sum_{r \in P_H} |r|^{-6} \\ \leq \; &C N^{11\k/2-2+2\e} + C N^{10\k-5+3\e} \leq C N^{5\k/2-\eps}  
\end{split} \]
where we used Lemma \ref{lm:nu-norms}, (\ref{eq:intVeta}), the assumption $3\k -2 + 4\e < 0$ and the remark that $|\eta_{r+v}| \leq CN^\kappa |r|^{-2}$, for all $r \in P_H$ and $v \in P_S$. We can proceed similarly to estimate $X_2$.  In the second term on the r.h.s. of (\ref{eq:thetam-one}), we have to sum over $(m-1) (m-2)/2$ pairs of indices $j,k$. With $\theta\big( \{r_j, v_j \}_{j=1}^{m-1} \big) \leq  \theta\big( \{r_j, v_j \}_{j=1}^{m-3} \big)$ and again with Lemma \ref{lm:nu-norms} and(\ref{eq:intVeta}), we arrive at 
\[ \begin{split} \frac{\text{X}_2}{\| \xi_\nu \|^2}  \leq\; & \frac{C}{N^3} \sum_{\substack{r , r' , r'' \in P_H, \\ v,v', v'' \in P_S}} N^\kappa |\widehat{V} (r/N^{1-\k}) |\eta_r + \eta_{r+v}| |\eta_{r'} + \eta_{r'+v'}|^2 |\eta_{r''} + \eta_{r''+v''}|^2 \sigma_v^2 \sigma_{v'}^2 \sigma_{v''}^2 \\ &\hspace{4.5cm} \times \Big[ \sum_{ \substack{p \in \{ -r , r+v \}, \\ p' \in \{ -r', r' + v'\}}} \delta_{p,-p'+v''}  + \sum_{\substack{p' \in \{-r', r'+v'\}, \\ p'' \in \{-r'', r'' + v''\}}} \delta_{v,p'+p''} \Big] \\
\leq \; &C N^{4\k - 3} \| \s_S \|^6 \| \eta_H \|^2  \sum_{r \in P_H} |r|^{-6} + C N^{6\k -3} \| \s_S \|^6  \sum_{r \in P_H} \frac{|\widehat{V} (r/N^{1-\k})|}{r^2} \sum_{r' \in P_H} |r'|^{-8}  \\ \leq \; &C N^{29\k/2 -7 + 5\eps} \leq C N^{5\k/2} \cdot N^{12 \k -7 + 5\eps}  \,. 
\end{split} \]
Thus, $|J_\cC| / \| \xi_\nu \|^2 \leq N^{5\k/2} \cdot \max \{ N^{-\eps} , N^{12\k -7 + 5\eps} \}$. With (\ref{eq:IC-fin}), this implies (\ref{eq:C}).


\subsection{Expectation of the quartic term}

In this section we show the bound \eqref{eq:V} for the expectation of $\cV_N^{(H)}$.
Pairing momenta in $P_S$, similarly as we did in (\ref{eq:normexpA2}) and in the previous subsections, we obtain 
\begin{equation}\label{eq:VH0} \begin{aligned}
	\bmedia{\xi_\nu,\cV_N^{(H)} \xi_\nu}=&\frac1{2N}\sum_{m\geq 1}\frac1{m!}\frac1{N^m} \sum_{\substack{v_1 \in P_S,\, r_1,  \tl r_1 \in P_H:\\ r_1+v_1,\, \tl r_1 + v_1 \in P_H}} \cdots\sum_{\substack{v_m \in P_S,\, r_m,  \tl r_m \in P_H:\\ r_m+v_m,\, \tl r_m + v_m \in P_H}}  \\
	&\times \theta\big( \{r_j, v_j \}_{j=1}^{m} \big)  \theta\big( \{\tl r_j, v_j \}_{j=1}^{m} \big) \prod_{i=1}^m\eta_{r_i}\eta_{\tilde{r}_i}\s_{v_i}^2\hspace{-0.3cm}\sum_{\substack{r\in \L^*,  p,q \in P_H: \\p+r,\,q+r\in P_H}}\hspace{-0.2cm} N^\k \widehat V(r/N^{1-\k})\\
	&\times\bmedia{\O,A_{r_1,v_1}\dots A_{r_m,v_m}\,a^*_{p+r}a^*_{q} a_p a_{q+r} A^*_{\tilde{r}_1,v_1}\dots A^*_{\tilde{r}_m,v_m}\O}
\end{aligned}\end{equation} 
where we use the notation $A_{r_i, v_i} = a_{r_i+v_i} a_{-r_i}$ that was already introduced in (\ref{eq:normexpA2}). Next we observe that, because of the cutoffs $\theta(\{r_j,v_j\}_{j=1}^m)$ and 
$\theta(\{\tl r_j,v_j\}_{j=1}^m)$, at most two indices $i,j \in \{ 1, \dots , m \}$ can be involved in contractions with the observable $a_{p+r}^* a_q^* a_p a_{q+r}$. We distinguish two possible cases. 
\begin{enumerate}[1)]
\item there exists an index $i \in \{ 1,\dots,m \}$ such that $a_p$, $a_{q+r}$ are contracted with $A^*_{\tl r_i,v_i}$ and $a^*_{q}$, $a^*_{p+r}$ are contracted with $A_{r_i,v_i}$
\item there are two indices $i \neq j \in \{ 1, \dots , m \}$ such that the operators $a_p$ and $a_{q+r}$ are contracted with $a^*_{\tl p_i}$ and $a^*_{\tl p_j}$ for some $\tl p_\ell \in \{-\tl r_\ell, \tl r_\ell + v_\ell\},$ $\ell=i,j$ and the operators $a^*_q, a^*_{p+r}$ are contracted with $a_{p_i}, a_{p_j}$, with $p_\ell\in\{-r_\ell,r_\ell+v_\ell\}, \ell=i,j$. Note that in this case the operators $a^*_{-\tl p_i+v_i},$ $a^*_{-\tl p_j+v_j}$ have to be contracted with $a_{-p_i+v_i},$ $a_{-p_j+v_j}.$
\end{enumerate}
We denote with $\text{V}_1$ and $\text{V}_2$ the contributions to $\bmedia{\xi_\nu,\cV_N^{(H)} \xi_\nu}$ arising from the two cases described above. Let us first consider $\text{V}_1$. There are $m$ choices (all leading to the same contribution) for the index $i \in \{ 1, \dots , m \}$ labelling momenta to be contracted with the observable. Let us fix $i = m$. Then we have $p=\tl p_m, q+r=-\tl p_m+v_m$ with $\tl p_m \in\{-\tl r_m,\tl r_m +v_m \}$, and $p+r=p_m, q=-p_m+v_m$ with $p_m\in \{-r_m,r_m+v_m\}$. Note that the choice of $p$ and $p+r$ also determines $q$ and $q+r$, since we always have $q = v_m - (p+r)$. The presence of the cutoffs immediately implies that $A_{r_j,v_j}$ is fully contracted with $A^*_{\tl{r}_j , v_j}$, for all $j \not = m$. We find 
\begin{equation}\label{eq:V11} 
\begin{split}
&\langle \xi_{\nu}, \text{V}_1 \xi_{\nu} \rangle \\
&=\frac1{2N}\sum_{m\geq 1}\frac1{(m-1)!}\frac1{N^m} \;\sum_{\substack{v_1 \in P_S\,, r_1, \tl r_1 \in P_H: \\ r_1 + v_1 ,\, \tl r_1 + v_1 \in P_H}} \cdots 
\sum_{\substack{v_{m} \in P_S\,, r_{m}, \tl r_{m} \in P_H: \\ r_{m} + v_{m} ,\, \tl r_{m} + v_{m} \in P_H}}\hskip -0.5cm\theta\big( \{r_j, v_j \}_{j=1}^{m} \big) \theta\big( \{\tl r_j, v_j \}_{j=1}^{m} \big)  \\[6pt]
&\, \times \hskip -0.1cm \prod_{j=1}^{m-1}\hskip -0.1cm\eta_{r_j}\eta_{\tl r_j}(\d_{r_j,\tl r_j}+\d_{-r_j,\tl r_j+v_j})\s_{v_j}^2\, \eta_{r_m} \eta_{\tl r_m}\s_{v_m}^2\hspace{-.5cm}
\sum_{\substack{r\in \L^*,\,p\in P_H:\\p-v_m,p+r\in P_H}} \hskip -0.7cm N^\k\widehat{V}(r/N^{1-\k}) \hskip -0.5cm \sum_{\substack{p_m\in \{-r_m,r_m+v_m\}\\\tl p_m\in \{-\tl r_m,\tl r_m+v_m\}}}\hspace{-0.8cm}\d_{p,p_m}\d_{p+r,\tl p_m} \,.
\end{split}\end{equation} 
Since here (in contrast to the previous subsections) the contraction does not fix $\tl{r}_m$ to be either $r_m$ or $-(r_m + v_m)$, we cannot erase the cutoff $\theta (\{ \tl{r}_j , v_j \}_{j=1}^m)$. With the decomposition (\ref{eq:remove}), we can replace, on the r.h.s. of (\ref{eq:V11}), 
\[ \theta\big( \{r_j, v_j \}_{j=1}^{m} \big) \theta\big( \{\tl r_j, v_j \}_{j=1}^{m} \big) = \theta \big( \{r_j, v_j \}_{j=1}^{m-1} \big)  \theta_m \big( \{ r_j, v_j \}_{j=1}^{m} \big) \theta_m  \big( \{\tl r_j, v_j \}_{j=1}^{m} \big). \]
Writing 
\[ \theta_m \big( \{\tl r_j, v_j \}_{j=1}^{m} \big) \theta_m  \big( \{\tl r_j, v_j \}_{j=1}^{m} \big) = 1 + \Big[ \theta_m \big( \{ r_j, v_j \}_{j=1}^{m} \big) \theta_m  \big( \{\tl r_j, v_j \}_{j=1}^{m} \big) - 1 \Big] \] 
we split (similarly as we did in the last subsection) $\langle \xi_\nu, \text{V}_1 \xi_\nu \rangle = I_\cV + J_\cV$, with  
\be \label{eq:I-cV-2}
I_\cV = \frac 1 {N^2} \sum_{r \in \L^*}  \sum_{\substack{ v\in P_S, p \in P_H:\\  p+r,p-v,\, p+r-v \in P_H}} \hskip -0.5cm  N^\k \widehat V(r/N^{1-\k})  \eta_{p} \big( \eta_{p+r} +\eta_{p+r-v} \big)  \s_{v}^2 \, \| \xi_\nu \|^2 
\ee
and 
\begin{equation}\label{eq:JV0} \begin{split}
J_\cV =&\frac1{2N}\sum_{m\geq 1}\frac1{2^{m-1} (m-1)!}\frac1{N^m}\sum_{\substack{v_1 \in P_S\,, r_1  \in P_H: \\ r_1 + v_1 \in P_H}} \cdots \sum_{\substack{v_{m-1} \in P_S\,, r_{m-1}  \in P_H: \\ r_{m-1} + v_{m-1} \in P_H}}  \sum_{\substack{v_{m} \in P_S \,, r_{m}, \tl r_{m} \in P_H: \\ r_{m} + v_{m} ,\, \tl r_{m} +  v_{m} \in P_H}}  \\
&  \times \theta\big( \{r_j, v_j \}_{j=1}^{m-1} \big) \big[ \theta_m(\{r_j, v_j\}_{j=1}^{m})\theta_m(\{ r^\sharp_j ,v_j\}_{j=1}^m)-1\big] \prod_{i=1}^{m-1} (\eta_{r_i} + \eta_{r_i + v_i})^2 \s_{v_i}^2 \, \\
&\times  \eta_{r_m} \eta_{\tl r_m} \s_{v_m}^2 \sum_{\substack{r\in \L^*}} N^\k \widehat V(r/N^{1-\k})\hspace{-0.5cm}  \sum_{\substack{ p \in P_H :\, p+r,\\ p-v_m,\, p+r-v_m\in P_H}} \sum_{\substack{p_m\in \{-r_m,r_m+v_m\}\\\tl p_m\in \{-\tl r_m,\tl r_m+v_m\}}}\hspace{-0.3cm}\d_{p,p_m}\d_{p+r,\tl p_m} 
\end{split}\end{equation} 
where $r^\sharp_j = r_j$ for $j=1,\dots, m-1$ and $r^\sharp_m = \tl{r}_m$ in the argument of $\theta_m$. Observing that, with Lemma \ref{lm:nu-norms} and (\ref{eq:intVeta}), 
\[ \begin{split} \frac 1 {N^2} \sum_{r \in \L^*} & \sum_{\substack{v \in P_S, p \in \L^* : \\ p\in P_H^c \text{ or } p-v \in P_H^c}}   N^\k |\widehat V(r/N^{1-\k})| |\eta_{p}|  \big( |\eta_{p+r}| + |\eta_{p+r-v}| \big)  \s_{v}^2  \\ &\leq C N^{-2+2\k} \| \s_S \|^2 \Big[ \sum_{|p| \leq N^{1-\k -\eps}} |p|^{-2} \Big]  \sup_{p\in \L^*} \sum_{r \in \L^*}   |\widehat V(r/N^{1-\k})| |\eta_{p+r}| \leq C N^{5\k/2 - \eps} \end{split} \]
we conclude from (\ref{eq:I-cV-2}) (switching $p+r \to p$ and $v \to -v$) that 
\begin{equation}\label{eq:IcV-fin} \frac{I_\cV}{\| \xi_\nu \|^2}  \leq \frac 1 {N^2} \sum_{\substack{v \in P_S, p  \in P_H: \\ p+v \in P_H}} \, \big( N^\k \widehat V(\cdot/N^{1-\k}) \ast \eta \big)_p  (\eta_p +\eta_{p+v})  \s_v^2  + C N^{5\k/2-\eps}. \end{equation} 

Let us now focus on the term $J_\cV$. With 
\begin{equation}\label{eq:W-dec} \begin{split} 
\Big| \theta_m \big(  &\{r_j, v_j \}_{j=1}^{m} \big) \theta_m\big( \{\tl r_j, v_j \}_{j=1}^{m} \big) -1 \Big| \\ \leq
\; &\sum_{j=1}^{m-1} \delta_{v_m, v_j} + \sum_{j=1}^{m-1}  \Big[ \sum_{\substack{p_j \in \{ - r_j, r_j + v_j \} \\ p_m \in \{ - r_m , r_m + v_m \}}} \delta_{p_m, p_j} + \sum_{\substack{p_j \in \{ - r_j, r_j + v_j \} \\ \tl{p}_m \in \{ - \tl{r}_m , \tl{r}_m + v_m \}}} \delta_{\tl{p}_m, p_j} \Big]  \\ &+ 
 \sum_{\substack{j,k=1 \\ j \not = k}}^{m-1}  \Big[ \sum_{\substack{p_j \in \{ - r_j, r_j + v_j \} \\ p_m \in \{ - r_m , r_m + v_m \}}} \delta_{p_m, -p_j + v_k}  + \sum_{\substack{p_j \in \{ - r_j, r_j + v_j \} \\ \tl{p}_m \in \{ - \tl{r}_m , \tl{r}_m + v_m \}}} \delta_{\tl{p}_m, -p_j + v_k} \Big] \\ &+ 
 \sum_{\substack{j,k=1 \\ j \not = k}}^{m-1}  \sum_{\substack{p_j \in \{ - r_j, r_j + v_j \} \\ p_k \in \{ - r_k , r_k + v_k \}}} \delta_{v_m , p_j + p_k} \end{split} \end{equation} 
we can bound $|J_\cV| \leq \text{W}_1 + \text{W}_2 + \text{W}_3 + \text{W}_4$, with $W_\ell$ indicating the contribution to (\ref{eq:JV0}) arising from the $\ell$-th term, on the r.h.s. of (\ref{eq:W-dec}). 

The term $\text{W}_1$ contains the sum of $(m-1)$ identical contributions, corresponding to $j \in \{ 1, \dots , m-1 \}$ in the first term on the r.h.s. of (\ref{eq:W-dec}). Let us fix $j=m-1$. Estimating $\theta\big( \{r_j, v_j \}_{j=1}^{m-1} \big) \leq \theta\big( \{r_j, v_j \}_{j=1}^{m-2} \big)$ and reconstructing the expression (\ref{eq:normA}) for $\| \xi_\nu \|^2$, we can bound (the momenta $r',r'',\tilde{r}''$ correspond to $r_{m-1}, r_m, \tilde{r}_m$)
\[ 
\frac{\text{W}_1}{\| \xi_\nu \|^2} \leq C N^{-3} \sum_{r \in \L^*} N^\kappa |\widehat{V} (r/N^{1-\kappa})|    \hspace{-.3cm} \sum_{\substack{r', r'', \tl{r}'' \in P_H \\ v'  \in P_S}} \hspace{-.3cm}   (\eta_{r'} + \eta_{r'+v'})^2 |\eta_{r''}| |\eta_{\tl{r}''}| \s_{v'}^4 \hspace{-.4cm} \sum_{\substack{p'' \in \{ -r'', r''+v' \} \\ \tl{p}'' \in \{ -\tl{r}'', \tl{r}'' +v' \}}}  \hspace{-.5cm}  \delta_{p''+r ,\tl{p}''}.
\]
With Lemma \ref{lm:nu-norms} and with the estimate 
\begin{equation} \label{eq:Vetaeta}  \sup_{v \in P_S \cup \{ 0 \}} \frac{1}{N^2} \sum_{\substack{r \in \L^*, q \in P_H : \\ q-r \in P_H}} N^\kappa |\widehat{V} (r/N^{1-\k})| |\eta_{q-v} | |\eta_{q-r}|  \leq C N^{\k} \end{equation} 
which can be shown similarly to (\ref{eq:intVeta}) (using $V \in L^q (\bR^3)$, for some $q > 3/2$) we find 
\begin{equation}\label{eq:W1-fin} \frac{\text{W}_1}{\| \xi_\nu \|^2} \leq C N^{\kappa -1} \| \eta_H \|^2 \| \s_S \|_\infty^2 \| \s_S \|^2 \leq C N^{11\k/2 -2 + 2\e} \leq C N^{5\kappa/2-\eps}  \end{equation}
since $3\kappa - 2 + 4\eps < 0$.  Analogously, we bound, with (\ref{eq:intVeta}) and Lemma \ref{lm:nu-norms}
\begin{equation} \label{eq:W2-fin} \begin{split} \frac{\text{W}_2}{\| \xi_\nu \|^2} \leq \; &C N^{-3} \sum_{r \in \L^*} N^\kappa |\widehat{V} (r/N^{1-\kappa})|    \hspace{-.3cm} \sum_{\substack{r', r'', \tl{r}'' \in P_H \\ v', v''  \in P_S}} \hspace{-.3cm}   (\eta_{r'} + \eta_{r'+v'})^2 |\eta_{r''}| |\eta_{\tl{r}''}| \s_{v'}^2 \s_{v''}^2 \hspace{-.3cm}\\ &\times  \sum_{\substack{p'' \in \{ -r'', r''+v'' \} \\ \tl{p}'' \in \{ -\tl{r}'', \tl{r}'' +v'' \}}}  \hspace{-.3cm}  \delta_{p'' + r , \tl{p}''} \Big[ \sum_{\substack{p' \in 
\{ -r', r'+v' \} \\ p'' \in \{ -r'', r'' +v'' \}}}  \delta_{p', p''}  + \sum_{\substack{p' \in 
\{ -r', r'+v' \} \\ \tl{p}''  \in \{ -\tl{r}'', \tl{r}'' +v'' \}}}  \delta_{p', \tl{p}''}  \Big] 
\\ \leq \; &C N^{-3+5\k}  \| \s_S \|^4 \Big[\sum_{r' \in P_H} |r'|^{-6}\Big] \Big[ \sup_{r' \in \L^*} \sum_{r \in \L^* , r \not = -r'} \frac{|\widehat{V} (r/N^{1-\k})|}{|r+r'|^{2}} \Big] \leq C N^{5\k/2-\eps} .
\end{split} \end{equation} 
As for $\text{W}_3$, there are $(m-1)(m-2)$ possible choices of the indices $j,k$ in (\ref{eq:W-dec}), all leading to the same contribution. We fix $j = m-1$ and $k =m-2$. Estimating now $\theta\big( \{r_j, v_j \}_{j=1}^{m-1} \big) \leq \theta\big( \{r_j, v_j \}_{j=1}^{m-3} \big)$, we obtain, with (\ref{eq:intVeta}), 
\begin{equation}\label{eq:W3-fin} \begin{split} \frac{\text{W}_3}{\| \xi_\nu \|^2} \leq \; &C N^{-4} \sum_{r \in \L^*} N^\kappa |\widehat{V} (r/N^{1-\kappa})|    \\  &\times \sum_{\substack{r', r'', r''', \tl{r}''' \in P_H \\ v', v'', v'''  \in P_S}} \hspace{-.3cm}   (\eta_{r'} + \eta_{r'+v'})^2  (\eta_{r''} + \eta_{r''+v''})^2 |\eta_{r'''}| |\eta_{\tl{r}'''}| \s_{v'}^2 \s_{v''}^2 \s_{v'''}^2 \hspace{-.5cm}\\ &\times  \sum_{\substack{p''' \in \{ -r''', r'''+ v''' \} \\ \tl{p}''' \in \{ -\tl{r}'''  , \tl{r}''' + v''' \}}}  \hspace{-.5cm}  \delta_{p'''+r , \tl{p}'''}  \Big[ \sum_{\substack{p' \in 
\{ -r', r'+v' \} \\ p''' \in \{ -r''', r''' +v''' \}}} \hspace{-.5cm}  \delta_{p''', -p'+v''}  + \sum_{\substack{p' \in 
\{ -r', r'+v' \} \\ \tl{p}''' \in \{ -\tl{r}''', \tl{r}''' +v''' \}}} \hspace{-.3cm}  \delta_{\tl{p}''' , - p' + v''}  \Big] 
\\ \leq \; &C N^{-3+ 4\k} \| \s_S \|^6 \| \eta_H \|^2 \sum_{r' \in P_H} |r'|^{-6} \leq C N^{5\k/2} \cdot N^{12\k - 7 +4\eps} .
\end{split} \end{equation} 
Analogously, with Lemma \ref{lm:nu-norms} and (\ref{eq:Vetaeta}),  we find 
\[ \begin{split} \frac{\text{W}_4}{\| \xi_\nu \|^2} \leq \; &C N^{-4} \sum_{r \in \L^*} N^\kappa |\widehat{V} (r/N^{1-\kappa})|    \\  &\times \sum_{\substack{r', r'', r''', \tl{r}''' \in P_H \\ v', v'', v'''  \in P_S}} \hspace{-.3cm}   (\eta_{r'} + \eta_{r'+v'})^2  (\eta_{r''} + \eta_{r''+v''})^2 |\eta_{r'''}| |\eta_{\tl{r}'''} | \s_{v'}^2 \s_{v''}^2 \s_{v'''}^2 \hspace{-.5cm}\\ &\times  \sum_{\substack{p''' \in \{ -r''', r'''+ v''' \} \\ \tl{p}''' \in \{ -\tl{r}''' , \tl{r}'''+ v''' \}}}  \hspace{-.5cm}  \delta_{p''' + r , \tl{p}''' }  \sum_{\substack{p' \in 
\{ -r', r'+v' \} \\ p'' \in \{ -r'', r'' +v'' \}}} \hspace{-.5cm} \delta_{v''', p' +p''} 
\\ \leq \; &C N^{-2+5\k}\| \s_S \|^6 \sum_{r' \in P_H} |r'|^{-8} \leq C N^{5\k/2} \cdot N^{12\k -7 + 5\eps} .
\end{split} \]
Together with (\ref{eq:W1-fin}), (\ref{eq:W2-fin}), (\ref{eq:W3-fin}), we conclude that 
\begin{equation}\label{eq:JV-fin} | J_\cV | \leq C N^{5\k/2} \cdot \max \{ N^{-\eps} , N^{12\k - 7 + 5\eps} \}. \end{equation} 

Finally, we consider the term $\text{V}_2$, associated with the second case listed after (\ref{eq:VH0}). 
We fix $i =m$ and $j=m-1$ and we consider all possible contractions of $a_p$ with $a^*_{\tl p_m}$, of  $a_{q+r}$ with $a^*_{\tl{p}_{m-1}}$ and of $a_q^*, a_{p+r}^*$ with $a_{p_m}, a_{p_{m-1}}$, where $\tl{p}_\ell  \in \{-\tl r_\ell,\tl r_\ell+v_\ell\}$ and $p_\ell \in \{ -r_\ell , r_\ell + v_\ell \}$, for $\ell=m,m-1$. We obtain 
\[\begin{aligned} 
	&\bmedia{\xi_\nu,\text{V}_2\, \xi_\nu} \\
& = \frac{1}{2N} \sum_{m\geq 2}\frac1{(m-2)!}\frac1{N^m} \sum_{\substack{v_1 \in P_S\,, r_1, \tl r_1 \in P_H: \\ r_1 + v_1 ,\, \tl r_1 +  v_1 \in P_H}} \cdots 
\sum_{\substack{v_{m} \in P_S\,, r_{m}, \tl r_{m} \in P_H: \\ r_{m} + v_{m} ,\, \tl r_{m} +  v_{m} \in P_H}}\hskip -0.5cm\theta\big( \{r_j, v_j \}_{j=1}^{m} \big) \theta\big( \{\tl r_j, v_j \}_{j=1}^{m} \big)  \\
& \hskip 0.5cm \times \prod_{i=1}^{m-2}\eta_{r_i}\eta_{\tilde{r}_i} \big( \d_{\tl r_i,r_i}+\d_{-\tl r_i,r_i+v_i}\big)\s_{v_i}^2 \, \prod_{j=m,m-1} \eta_{r_j} \eta_{\tl r_j}  \s_{v_j}^2\\
		&\hspace{0.5cm}\times \sum_{\substack{r \in \L^*,\, p,q \in P_H: \\ p-r, q-r  \in P_H }} N^\k \widehat V(r/N^{1-\k})   \sum_{\substack{\tl p_\ell \in \{ -\tl r_\ell, \tl r_\ell+v_\ell\}  \\ \ell=m-1,m } } \d_{p, \tl p_m} \d_{q+r, \tl p_{m-1}} \\
& \hspace{0.5cm}  \times  \hskip -0.5cm 
\sum_{\substack{ p_\ell \in \{ - r_\ell,  r_\ell+v_\ell\}   \\ \ell=m-1,m } }  \big(\d_{q, p_m} \d_{p+r,  p_{m-1}} +  \d_{q,  p_{m-1}} \d_{p+r, p_{m}} \big)  \big( \d_{\tl p_m, p_m} + \d_{-\tl p_m+v_m, -p_{m-1}+v_{m-1}} \big) .
\end{aligned}\]
Estimating $\theta\big( \{r_j, v_j \}_{j=1}^{m} \big) \theta\big( \{\tl r_j, v_j \}_{j=1}^{m} \big) \leq \theta \big( \{r_j, v_j \}_{j=1}^{m-2} \big)$ and using Lemma \ref{lm:nu-norms} and the condition $3\k - 2 + 4\eps < 0$, we find 
\[\begin{aligned} \label{eq:Vb}
\frac{|\bmedia{\xi_\nu,\text{V}_2\, \xi_\nu}|}{\| \xi_\nu \|^2} \\ \leq \; & C N^{-3} \sum_{r \in \L^*} N^\kappa \widehat{V} (r/N^{1-\k}) \sum_{\substack{r',\tl{r}',r'', \tl{r}'' \in P_H \\ v',v'' \in P_S}} |\eta_{r'}| |\eta_{\tl{r}'}| |\eta_{r''}| |\eta_{\tl{r}''}| \s_{v'}^2 \s_{v''}^2 \\ &\times 
\sum_{\substack{p' \in \{ -r' , r' + v' \} \\ p'' \in \{ -r'' , r'' + v'' \}}} 
\sum_{\substack{\tl{p}' \in \{ -\tl{r}' , \tl{r}' + v' \} \\ \tl{p}'' \in \{ -\tl{r}'' , \tl{r}'' + v'' \} }} (\delta_{\tl{p}', p''+r} \delta_{\tl{p}'' + r, p'} + \delta_{\tl{p}', p'+r} \delta_{\tl{p}'' + r, p''}) \\ &\hspace{6cm} \times  (\delta_{\tl{p}'', p''} + \delta_{-\tl{p}'' + v'', -p'+v'}) \\ 
\leq  \; &C N^{-3+\k} \| \eta_H \|^4 \| \s_S \|^4 \leq C N^{10\k - 5 + 2\eps} \leq C N^{5\k/2 -\eps}\,.
\end{aligned} \]
With (\ref{eq:IcV-fin}) and (\ref{eq:JV-fin}), we obtain (\ref{eq:V}).


\appendix

\section{Proof of Proposition \ref{prop:localization}}\label{app:localization}

The proof of Prop. \ref{prop:localization} is based on standard results, which are collected in this section for the reader convenience. In particular we follow \cite{R} (see Lemma 2.1.3) and \cite[Sec. 12]{YY} for  Lemma \ref{lm:Dir} and \ref{lm:Ltilde} and the proof of Lemma 3.3.2 in  \cite{Aaen} for Lemma \ref{lm:GCtoC} (control on the second moment of $\cN$ allows us to avoid  the condition imposed in \cite{Aaen} that $V$ is strictly positive around the origin).

The proof of Prop. \ref{prop:localization} is divided in three parts. First, we show how to switch from periodic boundary conditions to Dirichlet boundary conditions, increasing a bit the size of the box. In the second step, we replicate the Dirichlet trial state obtained in the first step, to obtain an upper bound on the energy in a sequence of boxes, whose size increases to infinity (but with fixed density). In the last step, we show how to pass from the grand canonical to the canonical setting.  

Let $\Psi_L = \{ \Psi_L^{(n)} \}_{n \geq 0} \in \cF (\L_L)$ be a normalized trial state for the Fock-space Hamiltonian $\cH$ defined on the box $\L_L$ with periodic boundary conditions (in fact, we denote by $\Psi_L^{(n)} (x_1, \dots, x_n)$ the $L$-periodic extension of $\Psi_L^{(n)}$ to the whole space $\bR^{3n}$). For $u \in \L_L$, we define $\Psi_{L+2\ell,u}^D \in \cF (\L_{L+2\ell}^u)$, where $\L_{L+2\ell}^u=u+\L_{L+2\ell}$ is a box centered at $u$, with side length $L+2\ell$, setting, for any $n \in \bN$,
\be\label{eq:PsiDir}
	(\Psi_{L+2\ell,u}^{\mathrm{D}})^{(n)}(x_1, \dots , x_n) = \Psi_L^{(n)} (x_1, \dots , x_n) \prod_{i=1}^n Q_{L,\ell} (x_i - u)
\ee
where $Q_{L,\ell} (x_i) = \prod_{j=1}^3 q_{L,\ell} (x_i^{(j)})$ with $q_{L,\ell} : \bR \to [0;1]$ defined by 
\[
	q_{L,\ell}(t)=\begin{cases}
	\cos\big(\frac{\pi(t+L/2-\ell)}{4\ell}\big) \qquad &\text{if $\big|t+\frac{L}{2}\big|\leq \ell$}\\
	1&\text{if $|t|<\frac{L}{2}-\ell$}\\
	\cos\big(\frac{\pi(t-L/2+\ell)}{4\ell}\big) &\text{if $\big|t-\frac L2\big|\leq \ell$}\\
	0 &\text{otherwise\,.}
	\end{cases}
\]
By definition $(\Psi_{L+2\ell,u}^{\mathrm{Dir}})^{(n)}$ satisfies Dirichlet boundary condition on the box 
$\L_{L+2\ell}^u$. The following lemma allows us to compare energy and moments of the number of particles of $\Psi^\text{D}_{L+2\ell , u}$ with those of $\Psi_L$. 
\begin{lemma}\label{lm:Dir}
Under the assumptions of Prop. \ref{prop:localization}, let $\Psi_{L+2\ell, u}^\text{D}$ be defined as in 
\eqref{eq:PsiDir} with $u\in \L_L$. Then we have  $\|	\Psi_{L+2\ell,u}^{\mathrm{D}}\|=1$. Moreover for all $j \in \bN$ 
\[\label{eq:Ndir}
\bmedia{\Psi_{L+2\ell,u}^{\mathrm{D}},\cN^j \Psi_{L+2\ell,u}^{\mathrm{D}}}=\bmedia{\Psi_L, \cN^j \Psi_L} \,,
\]
and there exists $\bar{u}\in\L_{L}$ such that
\be\label{eq:enDir}
\bmedia{\Psi_{L+2\ell,\bar u}^{\mathrm{D}},\cH\,\Psi_{L+2\ell,\bar u}^{\mathrm{D}}}\leq \bmedia{\Psi_L,\cH\,\Psi_L}+\frac{C}{L\ell}\bmedia{\Psi_L,\cN \Psi_L}
\ee
for a universal constant $C > 0$. 
\end{lemma}
\begin{proof}
For an arbitrary $L$-periodic function $\psi\in L^2_{\mathrm{loc}}(\bR)$, we find 
\be\label{eq:|psiq|}
\int_{-\frac L2-\ell}^{\frac L2+\ell}dt |\psi(t)|^2 q(t)^2 = \int_{-\frac L2}^{\frac L2}dt |\psi(t)|^2.
\ee
To prove (\ref{eq:|psiq|}), we combine (using the periodicity of $\psi$) the integral over $[-L/2-\ell ; -L/2]$ with the integral over $[L/2 - \ell ; L/2]$ and the integral over $[-L/2 ; - L/2 + \ell ]$ with the integral over $[L/2 ; L/2 +\ell]$ (using that $\cos^2 x + \cos^2 (x - \pi /2) = 1$). 

Applying \eqref{eq:|psiq|} (separately on each variable), we obtain that $\|  (\Psi_{L+2\ell,u}^\text{D} )^{(n)} \| = \| \Psi_L^{(n)} \|$, for all $n \in \bN$. This implies that $\| \Psi_{L+2\ell,u}^\text{D} \| = \| \Psi_L \| = 1$ and that $\langle \Psi_{L+2\ell,u}^\text{D} , \cN^j \Psi_{L+2\ell , u}^\text{D} \rangle = \langle \Psi_L, \cN^j \Psi_L \rangle$ for all $j \in \bN$. 

To compute the expectation of the kinetic energy in the state $\Psi_{L+2\ell,u}^{\mathrm{D}}$, we observe that, for any $L$-periodic $\psi \in L^2_{\mathrm{loc}}(\bR)$ with $\psi' \in  L^2_{\mathrm{loc}}(\bR)$, we have (since $\psi'$ is also $L$-periodic) 
\[
\int_{-\frac L2-\ell}^{\frac L2+\ell}dt \, |(q\psi)'(t)|^2=\int_{\frac L2}^{\frac L2}dt \, |\psi' (t)|^2 + \int_{-\frac L2-\ell}^{\frac L2+\ell}dt \, \big[ |\psi(t)|^2 q'(t)^2 + q(t) q'(t) \frac{d}{dt}|\psi(t)|^2\big]
\]
	where we used periodicity of $\psi'$ and \eqref{eq:|psiq|}. Integrating by parts and using $q(\pm (L/2+\ell)) = q' (\pm (L/2 - \ell)) = 0$, we get 
\be\label{eq:nabla_psiq}\begin{aligned}
		\int_{-\frac L2-\ell}^{\frac L2+\ell}dt \, |(q\psi)'(t)|^2= &\int_{-\frac L2}^{\frac L2} dt \, |\psi'(t)|^2-\int_{-\frac L2-\ell}^{\frac L2+\ell}dt \, |\psi(t)|^2 \, q(t)q''(t) \\
			\leq & \int_{-\frac L2}^{\frac L2}dt \,  |\psi'(t)|^2 + \frac C {\ell^2} \int_\bR dt \, 
			|\psi(t)|^2 \chi_{L,\ell}(t)
	\end{aligned}\ee
	where $\chi_{L,\ell}(t)=\chi_\ell(t+L/2)+\chi_\ell(t-L/2)$ with $\chi_r(t)$ the characteristic function of $[-r,r]$ and we used $|q'' (t)|\leq C \ell^{-2} \chi_{L,\ell}(t)$. Applying (\ref{eq:nabla_psiq}) (separately in every direction), we obtain
\begin{equation}\label{eq:kin-loc} \begin{split}  
\| \nabla_{x_j} &(\Psi^\text{D}_{L+2\ell, u})^{(n)} \|^2 \\ \leq \; & \| \nabla_{x_j} (\Psi^\text{D}_{L+2\ell,u})^{(n)} \|^2 \\ &+ \frac{C}{\ell^2} \int_{\bR^3} dx_j\, \wt{\chi}_{L,\ell} (x_j - u) \int_{\L_L^{n-1}} dx_1 \dots dx_{j-1} dx_{j+1} \dots dx_n \; |\Psi_L^{(n)} (x_1, \dots , x_n) |^2  \end{split} \end{equation} 
where we defined $\widetilde{\chi}_{L,\ell} (x) =\sum_{k=1}^3 \chi_{L,\ell} (x^{(k)}) \prod_{\substack{j\neq k}}^3 \chi_{\frac L2}(x^{(j)})$. 

To compute the potential energy of $\psi_L$, we have to consider the $L$-periodic extension $V_L (x)  
= \sum_{m\in \bZ^3}V(x+mL)$ of $V$. Since we assumed $V$ to be positive and supported in $B_R (0)$ and that $L>R$, we get $V(x)\leq V_L(x)$ which implies that, for any $i \not = j$,  
\[ \begin{split} 
|(\Psi_{L+2\ell,u}^{\mathrm{D}})^{(n)}(x_1, &\dots , x_n)|^2 \, V(x_i-x_j) \\ &\leq \Big| \Psi_L^{(n)} (x_1, \dots , x_n) \, \sqrt{V_L (x_i-x_j)} \Big|^2 \, \prod_{k=1}^n Q_{L,\ell} (x_k-u)^2.
\end{split} \]
Applying \eqref{eq:|psiq|}, we obtain 
\be\label{eq:pot_Dir}\begin{aligned}
\int_{(\L_{L+2\ell}^u)^n} dx_1 \dots dx_n \, &|(\Psi_{L+2\ell,u}^{\mathrm{D}})^{(n)}(x_1, \dots , x_n)|^2  V (x_i-x_j) \\ &\leq \int_{\L^n_L} dx_1 \dots dx_n \,  |\Psi^{(n)}_L (x_1, \dots , x_n)|^2 \, V_L (x_i-x_j).
	\end{aligned}\ee

From (\ref{eq:kin-loc}) and (\ref{eq:pot_Dir}), we conclude (using the bosonic symmetry) 
\begin{multline*}
\bmedia{\Psi_{L+2\ell,u}^{\mathrm{D}},\cH \Psi_{L+2\ell,u}^{\mathrm{D}}}\\
		\leq\bmedia{\Psi_L,\cH\Psi_L}+\frac{C}{\ell^2}\sum_{n\geq0}n \int_{\bR^3} dx_1 \, \widetilde{\chi}_{L,\ell}(x_1-u) \int_{\L_L^{n-1}} dx_2\dots dx_n \, |\Psi_L^{(n)}(x_1,\dots,x_n)|^2.
\end{multline*}
Averaging over $u\in \L_L$ we conclude (since $\| \wt{\chi}_{L,\ell} \|_1 \leq C L^2 \ell$) 
\[\int_{\L_L}du\,\bmedia{\Psi_{L+2\ell,u}^{\mathrm{D}},\cH\, \Psi_{L+2\ell,u}^{\mathrm{D}}}\leq L^3\bmedia{\Psi_L,\cH\,\Psi_L}+\frac{C L^2}{\ell}  \langle \Psi_L , \cN \Psi_L \rangle\,.  \]
Hence, there exists $\bar{u}\in \L_L$ so that  \eqref{eq:enDir} holds.

\end{proof}

From now on, let us define $\Psi_{L+2\ell}^\text{D} \in \cF (\L_{L+2\ell})$, setting $(\Psi_{L+2\ell}^\text{D})^{(n)} (x_1, \dots , x_n) = (\Psi_{L+2\ell , \bar{u}}^\text{D})^{(n)} (x_1 - \bar{u}, \dots , x_n - \bar{u})$, with $\Psi_{L+2\ell,\bar{u}}^\text{D}$ from Lemma \ref{lm:Dir}. Since $\Psi_{L+2\ell}^\text{D}$ satisfies Dirichlet boundary conditions, we can replicate it into several adjacent copies of $\L_{L+2\ell}$,  separated by corridors of size $R$ (to avoid interactions between different boxes). This allows us to construct a sequence of trial states on boxes, with increasing volume (but keeping the density fixed).  

Let $t \in \bN$ and $\tl{L} = t (L+2\ell +R)$. We think of the large box $\L_{\tl{L}}$ as the (almost) disjoint union of $t^3$ shifted copies of the small box $\L_{L+2\ell+R}$, centered at \begin{equation}\label{eq:ci} (-\tl{L}/2 , - \tl{L}/2, - \tl{L}/2) + (L+2\ell +R) \cdot (i_1 - 1/2 , i_2 - 1/2, i_3 - 1/2)  \end{equation} 
with $i_1, i_2, i_3 \in \{ 1, \dots , t \}$. Let $\{ c_i \}_{i=1}^{t^3}$ denote an enumeration of the centers (\ref{eq:ci}). We define $\Psi^\text{D}_{\tl{L}} \in \cF (\L_{\tl{L}})$ by setting 
\begin{equation}\label{eq:Psi-large} 
(\Psi_{\tl{L}}^\text{D} )^{(m)} (x_1, \dots , x_{m}) = \frac{1}{\| (\Psi^\text{D}_{L+2\ell})^{(n)} \|^{t^3 -1}}  \prod_{i=1}^{t^3} (\Psi_{L+2\ell}^\text{D})^{(n)} (x_{(i-1)n + 1} - c_i , \dots , x_{i n} - c_i) \end{equation} 
if $m = n t^3$ for an $n \in \bN$, and $(\Psi_{\tl{L}}^\text{D} )^{(m)} = 0$ otherwise (here we set $(\Psi_{L+2\ell}^\text{D})^{(n)} = 0$ if one of its arguments lies outside $\L_{L+2\ell}$). More precisely, $(\Psi_{\tl{L}}^\text{D} )^{(m)}$ should be defined as the symmetrization of (\ref{eq:Psi-large}) (but we can use (\ref{eq:Psi-large}) to compute the expectation of permutation symmetric observables). 

\begin{lemma}\label{lm:Ltilde}
Under the assumptions of Prop. \ref{prop:localization}, let $\Psi_{\tl{L}}^\text{D}$ be defined as above. Then $\| \Psi_{\tl{L}}^\text{D} \| = 1$,  
\[\label{eq:NGC}
\bmedia{\Psi_{\tilde{L}}^\mathrm{D},\cN^j\Psi_{\tilde{L}}^\mathrm{D}}=t^{3j}\bmedia{\Psi_{L+2\ell}^\mathrm{D},\cN^j \Psi_{L+2\ell}^\mathrm{D}} \]
for all $j \in \bN$, and 
\begin{equation}\label{eq:enGC}
\bmedia{\Psi_{\tilde{L}}^\mathrm{D},\cH\Psi_{\tilde{L}}^\mathrm{D}} = t^3 \bmedia{\Psi_{L+2\ell}^{\mathrm{D}}, \cH \Psi_{L+2\ell}^{\mathrm{D}}}.
\end{equation}
\end{lemma}
\begin{proof}
From the definition (\ref{eq:Psi-large}), we have $\| (\Psi_{\tl{L}}^\text{D})^{(nt^3)} \| =  \| (\Psi^\text{D}_{L+2\ell})^{(n)} \|$ for all $n \in \bN$. Since $(\Psi_{\tl{L}}^\text{D})^{(m)} = 0$, if $m \not = n t^3$, we conclude that $\| \Psi_{\tl{L}}^\text{D} \| = \| \Psi_{L+2\ell}^\text{D} \| = 1$ and also that, for $j \in \bN$,
\[
\bmedia{\Psi_{\tilde{L}}^\mathrm{D}, \cN^j \Psi_{\tilde{L}}^\mathrm{D}} \hskip -0.05cm=\sum_{n\geq 0} (t^3 n)^j \|(\Psi_{\tilde{L}}^\mathrm{D})^{(t^3n)}\|^2 \hskip -0.05cm=t^{3j} \sum_{n\geq 0} n^j \|(\Psi_{L+2\ell}^{\mathrm{D}})^{(n)}\|^2 \hskip -0.05cm= t^{3j} \bmedia{\Psi_{L+2\ell}^{\mathrm{D}},\cN^j \Psi_{L+2\ell}^{\mathrm{D}}}.
	\]

To prove (\ref{eq:enGC}), we observe, first of all, that for any $i=1, \dots , n t^3$,when the operator 
$\nabla_{x_i}$ acts on $(\Psi^\text{D}_{\tl{L}})^{(n t^3)}$, it only hits one of the factor $(\Psi^\text{D}_{L+2\ell})^{(n)}$ on the r.h.s. of (\ref{eq:Psi-large}). Similarly, for any $i,j \in \{ 1, \dots , m \}$, the operator $V(x_i - x_j)$ has non-zero expectation in the state $(\Psi^\text{D}_{\tl{L}})^{(n t^3)}$ only 
if $x_i, x_j$ are arguments of the same factor $(\Psi^\text{D}_{L+2\ell})^{(n)}$ on the r.h.s. of (\ref{eq:Psi-large}) (this observation is exactly the reason for introducing corridors of size $R$ 
between the small boxes, where the wave function vanishes). We conclude that $\bmedia{\Psi_{\tilde{L}}^\mathrm{D},\cH \Psi_{\tilde{L}}^\mathrm{D}} = t^3 \bmedia{\Psi_{L+2\ell}^{\mathrm{D}}, \cH \Psi_{L+2\ell}^{\mathrm{D}}}$, as claimed.
\end{proof} 
	
Finally, in Lemma \ref{lm:GCtoC}  we show how to obtain an upper bound for the ground state energy per particle in the canonical ensemble, starting from a trial state in the grand-canonical setting. Recall the notation $E(N,L)$ for the ground state energy of the Hamiltonian \eqref{eq:HN-0}, describing $N$ particles in the box $\L_L$, with Dirichlet boundary conditions. For $\r > 0$ with $\r L^3 \in \bN$, we introduce the notation 
\[ e_{L}(\r)=\frac{E(\r  L^3, L)}{ L^3}\,. \]
Comparing with the definition (\ref{eq:e-rho}), we find $e (\rho) = \lim_{L \to \infty} e_L (\rho)$ (where the limit has to be taken along sequences of $L$, with $\rho L^3 \in \bN$). 
 In the proof of Lemma \ref{lm:GCtoC}  we use the existence of the thermodynamic limit of the specific energy and its convexity (see \cite[Thm. 3.5.8 and 3.5.11]{Ruelle}), together with the following result on the Legendre transform of convex functions.

\begin{lemma} \label{lm:legendre}
Let $D \subset \bR$ be a closed interval, $f : D \to \bR$ be convex and continuous (also at the boundary of $D$). We define the Legendre transform $f^* : \bR \to \bR$ of $f$ by 
\be \label{eq:LegendreT}
 f^* (y) = \sup_{x \in D} \left[ x y - f(x) \right] 
\ee
Then $f^*$ is well-defined (because, by continuity, $x \to xy - f (x)$ is bounded on $D$, for all $y \in \bR$) and, for all $x \in D$, 
\begin{equation}\label{eq:lege1} f (x) = \sup_{y \in \bR}  \left[ xy - f^* (y) \right] \, . \end{equation} 
\end{lemma} 
\begin{proof}
By definition of $f^*$, we have $f^* (y) \geq xy - f(x)$ for all $x \in D, y \in \bR$. This implies that $f(x) \geq xy - f^* (y)$ for all $x \in D, y \in \bR$ and therefore that 
\be\label{eq:Legendre_bound} f(x) \geq \sup_{y \in \bR}  \left[ xy - f^* (y) \right] \ee
for all $x \in D$. On the other hand, fix $x_0 \in D$ and $t \leq f(x_0)$. Then, by convexity of $f$ (and by its continuity at the boundaries of $D$), we find a line through $(x_0,t)$ lying below the graph of $f$. In other words, there exists $y \in \bR$ such that $f(x) \geq t + y (x-x_0)$ for all $x \in D$. Thus $y x_0 - t \geq y x - f(x)$ for all $x \in D$, which implies that 
\[ y x_0 -t \geq f^* (y) \]
and therefore that $t \leq y x_0 - f^* (y)$. In particular, $t \leq \sup_{y \in \bR} [ y x_0 - f^* (y)]$. Since $t \leq f(x_0)$ was arbitrary, we conclude that $f(x_0) \leq \sup_{y \in \bR} [ y x_0 - f^* (y)]$. With (\ref{eq:Legendre_bound}) , we obtain that $f(x) = \sup_{y \in \bR}  \left[ xy - f^* (y) \right]$ for all $x \in D$. 
\end{proof}

\begin{lemma}\label{lm:GCtoC} 
Under the assumptions of Prop. \ref{prop:localization}, fix $\rho > 0$  and suppose that there exists a sequence $\Psi_L^\text{D} \in \cF (\L_L)$ (parametrized by $L$ with $\rho L^3 \in \bN$), satisfying Dirichlet boundary conditions, such that 
\begin{equation}\label{eq:cN-lower}  \langle \Psi_L^\text{D} , \cN \Psi^\text{D}_L \rangle \geq  \r ( 1 + c' \r) L^3 \,, \qquad \langle \Psi_L^\text{D} , \cN^{2} \Psi^\text{D}_L\rangle \leq C' (\r L^3)^2  \,.\end{equation} 
for some constants $c', C' > 0$. Then we have 
\[ e (\rho) \leq \lim_{L \to \infty} \frac{\langle \Psi_L^\text{D}, \cH \Psi_L^\text{D} \rangle}{L^3} \,.\]
\end{lemma} 

\begin{proof}
Using positivity of $\cH$ , we have, for any $\mu \geq 0$ and $M>0$,  
	\be\label{eq:GCtoC}\begin{aligned}
		&\frac{\bmedia{\Psi_{L}^\mathrm{D},\cH \Psi_{L}^\mathrm{D}}}{L^3} \\
& \quad \geq \frac{\mu}{L^3}\bmedia{\Psi_{L}^{\mathrm{D}},\cN\Psi_{L}^{\mathrm{D}}} +\frac{\bmedia{\Psi_{L}^{\mathrm{D}},(\cH-\mu\cN)\chi(\cN\leq M L^3) \Psi_{L}^{\mathrm{D}}}}{{L}^3} 
-\frac{\mu}{L^3}\bmedia{\Psi_{L}^{\mathrm{D}},\cN\chi(\cN>M L^3) \Psi_{L}^{\mathrm{D}}}\\
& \quad  \geq \frac{\mu}{L^3}\bmedia{\Psi_{L}^{\mathrm{D}},\cN\Psi_{L}^{\mathrm{D}}}+\sum_{m=0}^{M L^3}\bigg(e_{L}\bigg(\frac{m}{{L}^3}\bigg)-\mu \frac{m}{{L}^3}\bigg) \Big\|(\Psi_{{L}}^{\mathrm{D}})^{(m)}\Big\|^2
- \frac{\mu}{M L^6}\bmedia{\Psi_{L}^{\mathrm{D}},\cN^2\Psi_{L}^{\mathrm{D}}} \,,
\end{aligned}\ee
where we used the inequality $\chi(\cN> M L^3) \leq \cN/ (M L^3)$. 
Hence, with \eqref{eq:cN-lower} and fixing $M$ large enough (depending on $c', C'$) we find
\be\label{eq:GCbound}\begin{aligned}
\frac{\bmedia{\Psi_{L}^\mathrm{D},\cH \Psi_{L}^\mathrm{D}}}{L^3} & \geq  \mu \r  +\sum_{m=0}^{M L^3}\bigg(e_{L}\bigg(\frac{m}{{L}^3}\bigg)-\mu \frac{m}{{L}^3}\bigg) \Big\|(\Psi_{{L}}^{\mathrm{D}})^{(m)}\Big\|^2 \,.
\end{aligned}\ee
Next, we claim that 
\be\label{eq:copies}
e_{L}(\r)\geq \Big(1+\frac R{L}\Big)^3e\bigg(\r\Big(1+\frac{R}{L}\Big)^{-3}\bigg)\,.
\ee
Indeed, starting from an arbitrary normalized trial state $\psi$ describing $N = \r L^3$ particles in a box of side length $L$, with Dirichlet boundary conditions, we can construct, for any $r \in \bN$, a trial state describing $N'= N r^3  = \r L^3 r^3$ particles in a box of side length $L'= r (L+R)$, again with Dirichlet boundary conditions, by placing $r^3$ copies of the state $\psi$ in adjacent boxes and using that (thanks to the corridors of size $R$ between the boxes) particles in different boxes do not interact. 
This construction is very similar to the one presented around Lemma \ref{lm:Ltilde} (the difference is that here we work in the canonical setting, which makes things slightly simpler). Since $N' = [\r / (1 +R/L)^3 ] L'^3$, optimizing the choice of $\psi$, we obtain that $E ([\r / (1 + R/L)^3] L'^3 , L') \leq r^3 E (\rho L^3, L)$ and therefore that 
\[ e_{L'} (\r/(1 + R/L)^3) \leq e_L (\rho) / (1+R/L)^3 \,.\]
Taking the limit $L' \to \infty$ (along the sequence $L' = r (L+R)$, $r \in \bN$), we obtain (\ref{eq:copies}). 
Then \eqref{eq:GCbound} and \eqref{eq:copies} yield
\[
	\frac{\bmedia{\Psi_{L}^\mathrm{D},\cH \Psi_{L}^\mathrm{D}}}{L^3}\geq\;\mu\r-\Big(1+\frac RL\Big)^3e^*(\mu).
\]
where $e^*$ denotes the Legendre transform of $e : D \to \bR$, defined on the domain $D=[0,M]$, as in  \eqref{eq:LegendreT} (here we use the convexity of the specific energy $e$). It follows that 
	\[
		\lim_{L\to+\infty}\frac{\bmedia{\Psi_{L}^\mathrm{D},\cH \Psi_{L}^\mathrm{D}}}{L^3}\geq\;\mu\r-e^*(\mu) 
	\]
for all $\mu  \geq 0$. Thus 
\[
\lim_{L\to+\infty}\frac{\bmedia{\Psi_{L}^\mathrm{D},\cH \Psi_{L}^\mathrm{D}}}{L^3}\geq \; \sup_{\mu \geq 0} \Big[ \mu\r-e^*(\mu) \Big] = \sup_{\mu \in \bR} \Big[ \mu\r-e^*(\mu) \Big] = e (\rho) \]
where we used the fact that $e^* (0) = 0$ (because $e (\rho) \geq 0$ for all $\rho \geq 0$ and $e (0) = 0$) and $e^* (\mu) \geq -e (0) = 0$ for all $\mu \in \bR$ in the second step and Lemma \ref{lm:legendre} in the third step. 
\end{proof}

With Lemmas \ref{lm:Dir}, \ref{lm:Ltilde} and  \ref{lm:GCtoC} we are ready to show Prop. \ref{prop:localization}.  

\begin{proof}[Proof of Prop. \ref{prop:localization}] 
Given a normalized $\Psi_L \in \cF (\L_L)$ satisfying periodic boundary conditions with 
\[ 
\langle \Psi_L, \cN \Psi_L \rangle \geq \rho (1+c' \r) (L + 2\ell + R)^3 \,, \hskip 0.5cm  \langle \Psi_L, \cN^2 \Psi_L \rangle \leq C' \rho^2 (L + 2\ell + R)^6  \]
we find with Lemma \ref{lm:Dir} a normalized $\Psi_{L+2\ell}^\text{D} \in \cF (\L_{L+2\ell})$ satisfying Dirichlet conditions such that 
\[ 
  \langle \Psi_{L+2\ell}, \cN \Psi_{L+2\ell} \rangle \geq \rho (1+c' \r) (L + 2\ell + R)^3 \,, \hskip 0.5cm  \langle \Psi_{L+2\ell}, \cN^2 \Psi_{L+2\ell} \rangle \leq C' \rho^2 (L + 2\ell + R)^6  
\]
and 
\[ \langle \Psi_{L+2\ell}^\text{D} , \cH   \Psi_{L+2\ell}^\text{D} \rangle \leq \langle \Psi_L , \cH \Psi_L \rangle + \frac{C}{L\ell} \langle \Psi_L , \cN \Psi_L \rangle\,. \]
With Lemma \ref{lm:Ltilde}, we obtain a sequence $\Psi^\text{D}_{\tl{L}} \in \cF (\L_{\tl{L}})$, with $\tl{L} =t (L+2\ell+R)$ for $t \in \bN$, such that 
\[  \langle \Psi_{\tl{L}}^\text{D} , \cN   \Psi_{\tl{L}}^\text{D} \rangle \geq   \rho (1+ c' \r) \tl{L}^3 \,, \hskip 0.5cm   \langle \Psi_{\tl{L}}^\text{D} , \cN^2   \Psi_{\tl{L}}^\text{D} \rangle  \leq  C' \rho^2 \tl{L}^6  \]
and \[ \langle \Psi_{\tl{L}}^\text{D} , \cH   \Psi_{\tl{L}}^\text{D} \rangle 
\leq t^3 \langle \Psi_L , \cH \Psi_L \rangle + \frac{C t^3}{L\ell} \langle \Psi_L , \cN \Psi_L \rangle \,. \]
With Lemma \ref{lm:GCtoC}, we conclude that 
\[ \begin{split} 
e(\rho) &\leq \lim_{\tl{L} \to \infty}  \frac{\langle \Psi_{\tl{L}}^\text{D} , \cH   \Psi_{\tl{L}}^\text{D} \rangle}{\tl{L}^3} 
\\ & \leq \frac{1}{(1+2\ell/L + R/L)^3} \left[ \frac{\langle \Psi_L, \cH \Psi_L \rangle}{L^3} + \frac{C}{L^4 \ell} \langle \Psi_L, \cN \Psi_L \rangle \right] \\  &\leq  \frac{\langle \Psi_L, \cH \Psi_L \rangle}{L^3} + \frac{C}{L^4 \ell} \langle \Psi_L, \cN \Psi_L \rangle\,.
\end{split}  \] 
\end{proof}


\def\bskip{\\[-0.6cm]}

\end{document}